\documentclass[10pt,aps,amsfonts,prb,twocolumn,showpacs,superscriptaddress]{revtex4-1}
\usepackage{verbatim,epsfig,amsmath,amssymb,bm,epsf,graphicx,psfrag,bbold,amsthm,amsfonts}
\usepackage{hyperref}
\usepackage{enumitem}
\setlist{nosep}
\usepackage[all]{xy}
\usepackage{color}
\newtheorem{theorem}{Theorem}

\newcommand{\ket}[1]{|#1\rangle}
\newcommand{\bra}[1]{\langle #1|}
\newcommand{\innerproduct}[2]{\langle #1| #2\rangle}
\newcommand{\outerproduct}[2]{|#1\rangle\langle #2|}

\newcommand{\ztwo}{\mathbb{Z}_2}

\newcommand{\zthree}{\mathbb{Z}_3}
\newcommand{\ea}{\chi^{SG}_{\zthree}}

\newcommand{\thatis}{$i.e.~$}

\begin{document}

\title{Eigenstate phases with finite on-site non-Abelian symmetry}

\author{Abhishodh Prakash}
\affiliation{Department of Physics and Astronomy, State University of New York at
	Stony Brook, Stony Brook, NY 11794-3840, USA}
\affiliation{C. N. Yang Institute for Theoretical Physics, Stony Brook, NY 11794-3840, USA} 
\author{Sriram Ganeshan}
\affiliation{Simon's Center for Geometry and Physics, Stony Brook, NY 11794-3840, USA} 

\author{Lukasz Fidkowski}
\affiliation{Department of Physics and Astronomy, State University of New York at
	Stony Brook, Stony Brook, NY 11794-3840, USA}

\author{Tzu-Chieh Wei}
\affiliation{Department of Physics and Astronomy, State University of New York at
	Stony Brook, Stony Brook, NY 11794-3840, USA}
\affiliation{C. N. Yang Institute for Theoretical Physics, Stony Brook, NY 11794-3840, USA} 
	
\date{\today}

\begin{abstract}
{ We study the eigenstate phases of disordered spin chains with on-site finite non-Abelian symmetry. We develop a general formalism based on standard group theory to construct local spin Hamiltonians invariant under any on-site symmetry. We then specialize to the case of the simplest non-Abelian group, $S_3$, and numerically study a particular two parameter spin-1 Hamiltonian. We observe a thermal phase and a many-body localized phase with a spontaneous symmetry breaking (SSB) from $S_3$ to $\mathbb{Z}_3$ in our model Hamiltonian. We diagnose these phases using full entanglement distributions and level statistics.  We also use a spin-glass diagnostic specialized to detect spontaneous breaking of the $S_3$ symmetry down to $\mathbb{Z}_3$. Our observed phases are consistent with the possibilities outlined by Potter and Vasseur [Phys. Rev. B 94, 224206 (2016)], namely thermal/ ergodic and spin-glass many-body localized (MBL) phases. We also speculate about the nature of an intermediate region between the thermal and MBL+SSB regions where full $S_3$ symmetry exists.}
\end{abstract}

 \maketitle
 
\section{Introduction}
Statistical mechanics and thermodynamics are bridges that connect microscopic laws such as Newtonian and quantum mechanics to macroscopic phenomena that we measure in the laboratory. The validity of statistical physics relies on the existence of thermal equilibrium. For an isolated quantum system, the notion of thermalization is understood in the form of the \emph{Eigenstate Thermalization hypothesis} (ETH)~\cite{1991_Deutsch_ETH_PRA, 1994_Srednicki_ETH_PRE}. ETH posits that for a quantum system, an eigenstate embodies an ensemble and thermalization can be diagnosed by monitoring if subsystems are thermal with respect to the rest of the system. Furthermore, if the system thermalizes, all eigenstates are thermal. Integrable systems violate ETH due to the existence of an extensive number of conserved quantities that prevent the system from acting as a bath for itself. However, quantum integrable models are highly fine tuned and one recovers thermalization by any infinitesimal deviation from the integrable point.
 
Recently, many-body localization (MBL) has emerged as a \emph{generic} class of interacting and disordered isolated systems which violate ETH. Basko et al.~\cite{Basko06} showed that all many-body eigenstates remain localized to all orders in perturbation for an effective interacting disordered model. Several numerical works subsequently verified that all many-body eigenstates are localized in one dimensional disordered lattice models with short-range interactions~\cite{oganesyan2007, pal2010, bauerNayak2013completeset, vadim, bardarson2014, huse2015many}. Furthermore, there has been a mathematical proof by Imbrie~\cite{imbrie2014many} for the existence of MBL in a particular disordered spin model with short range interactions. The absence of thermalization has been further quantified as a consequence of emergent integrability due to the presence of a complete set of local integrals of motion (LIOMs). The key distinction from fine tuned integrable models is that the LIOMs or `lbits' (for localized bits) in the MBL phase are robust against perturbations. One can use these lbits to construct a phenomenological lbit Hamiltonian that captures the entanglement dynamics~\cite{serbyn2013, chandran2015, huseNandkishoreOganesyan2014phenomenology, ros2015integrals}.  
	
Having established the existence of MBL and its violation of ETH in certain models, natural questions that arise are ``what is the most robust version of MBL?" and ``does it lead to a  refined notion of ETH?" To this end, it is worthwhile to consider instabilities to the MBL phase that lead to delocalization and thereby the restoration of thermalization. Recent works have considered instabilities to the MBL phase due to a small bath~\cite{johri2015many, hyatt2017many}, external drive~\cite{lazarides2015fate}, Griffiths effects and dimensionality~\cite{de2016absence, de2017many}, topologically protected chiral edge~\cite{nandkishore2014marginal} and a single particle mobility edge~\cite{li2015many, modak2015many}. Contrary to the common wisdom that these instabilities would lead to the complete restoration of thermalization, preliminary numerical results have indicated that the lack of thermalization tends to survive in some form in all these cases. However, the fate of these exotic phases in the thermodynamic limit is still an open question. 
	
Potter and Vasseur~\cite{potter2016symmetry} have recently added another instability to this list. It was argued that the l-bit Hamiltonian `enriched' with non-Abelian symmetries is unstable to perturbations. The physical intuition of the delocalization and thermalization stems from the resonant splitting of extensive degeneracy in the many-body spectrum associated with the higher-dimensional irreducible representations (irreps) of the non-Abelian group. Any perturbation results in either resonant delocalization of the many-body spectrum and thermalization, a quantum critical glass (QCG)-like phase or a spin-glass (SG) phase by spontaneous symmetry breaking (SSB) to an Abelian subgroup.

Thus symmetry provides a platform to search for exotic violations of ETH beyond MBL in strongly interacting systems. To this end, in this paper, we develop a procedure to construct general Hamiltonians with global symmetries and analyze the thermalization and localization indicators for individual eigenstates. As a particular example, we construct a two-parameter Hamiltonian with on-site $S_3$ symmetry. Using numerical exact diagonalization of our model Hamiltonian, we calculate cut-averaged entanglement entropy (CAEE) distributions and level statistics which are indicators of localization and a spin-glass diagnostic which detects symmetry breaking. Within the accuracy of our numerical analysis, we are able to distinctly observe both a thermal phase and an MBL spin-glass (SG) phase with spontaneous symmetry breaking of $S_3$ to $\mathbb{Z}_3$ symmetry.

{  We also employ the same diagnostics to quantify an intermediate region between the aforementioned two phases where the full $S_3$ symmetry is intact. However, we cannot ascertain the fate of this region in thermodynamic limit, due to the possibility of quantum critical cone like finite-size effects~\cite{husekhemani2016critical}.} The paper is organized as follows. In Sec.~\ref{sec:model} we construct a general $S_3$ symmetric Hamiltonian using group theory methods. We numerically diagonalize our model and compute indicators of localization and symmetry breaking in ~\ref{sec:diagnostics}. We end the paper with discussion of results and conclusion. We provide a review of the conjecture by Potter and Vasseur~\cite{potter2016symmetry} on the incompatibility of MBL with non-Abelian symmetries and other details of our analysis in the Appendices.
 
\section{Model $S_3$ invariant Hamiltonian}
\label{sec:model}
	In this section, we analyze a specific spin-1 Hamiltonian that is invariant under the smallest non-Abelian group, $S_3$. 
	In terms of the spin angular momentum basis $|S=3,S_z=+1,0,-1\rangle$, the spin operators are
	\begin{eqnarray}
			S^x &=&~ \frac{1}{\sqrt{2}} \begin{pmatrix}
			0 & 1 & 0 \\
			1 & 0 & 1 \\
			0 & 1 & 0
			\end{pmatrix},~
			S^z = \begin{pmatrix}
			1 & 0 & 0 \\
			0 & 0 & 0 \\
			0 & 0 & -1
			\end{pmatrix}, \nonumber \\			
			S^y &=& \frac{1}{\sqrt{2}} \begin{pmatrix}
			0 & -i & 0 \\
			i & 0 & -i \\
			0 & i & 0
			\end{pmatrix}, \
			S^\pm = \frac{1}{\sqrt{2}} (S^x \pm i S^y). \nonumber
			\end{eqnarray}
	The symmetry group $S_3$ contains six elements: $S_3 = \{1,a,a^2,x,x a,x a^2\}$, and the two generators, $a$ and $x$, satisfy the properties $a^3=x^2=1$ and $x a x = a^{-1}$. Note that in this paper, we refer to the identity element of the group simply as $1$. 
	In the spin basis, they are chosen to have the following representations,
	\begin{eqnarray}
	V(a) = \begin{pmatrix}
	\omega & 0 & 0\\
	0 & 1 & 0\\
	0 & 0 & \omega^*
	\end{pmatrix},~
	V(x) = \begin{pmatrix}
	0 & 0 & 1\\
	0 & -1 & 0\\
	1 & 0 & 0
	\end{pmatrix},
	\end{eqnarray}
	where $\omega = e^{2 \pi i/3}$. It can be verified that the spin operators transform under the generators as follows,
		\begin{eqnarray}
	V(a)\,S^\pm\, V(a)^\dagger&=&\omega ^{\pm 1} S^\pm, \ V(a)\,S^z\,V(a)^\dagger=S^z\\
	V(x)\,S^\pm \,V(x)^\dagger&=&- S^\mp, \ V(x)\,S^z\,V(x)^\dagger=-S^z.
	\end{eqnarray}
	Using the symmetry arguments detailed in Appendix~\ref{app:building symmetric ham}, we construct the following Hamiltonian: 
	\begin{eqnarray}\label{eq:S3 Hamiltonian}
	H(\lambda, \kappa) &=& \lambda H_d(\kappa) +  H_t,\\
	H_d(\kappa) &=&  \sum_{i=1}^L (1-\kappa)~ h_i~ (S^z_i)^2 + \kappa~J_i~S^z_i S^z_{i+1}, \nonumber\\
	H_t &=& \Delta_t \left[H_a + H_b + H_c \right],\nonumber\\
	H_a &=& a\sum_{i=1}^{L}(S^+_i)^2 (S^-_{i+1})^2 + (S^-_i)^2 (S^+_{i+1})^2 +h.c,\nonumber\\
	H_b &=& b\sum_{i=1}^{L}(S^+_iS^z_i)(S^-_{i+1}S^z_{i+1}) + (S^-_{i}S^z_{i}) (S^+_{i+1}S^z_{i+1}) +h.c , \nonumber\\
	H_c &=& c\sum_{i=1}^{L}(S^+_i)^2 (S^+_{i+1}S^z_{i+1}) + (S^-_i)^2 (S^-_{i+1}S^z_{i+1}) +h.c. \nonumber
	\end{eqnarray}
	
	The above Hamiltonian consists of two parts: 1) The disordered part $H_d$ with a one body (disordered $h_i$ term) and a two-body l-bit term (disordered $J_i$ term). The two-body term is designed to drive spontaneous symmetry breaking of non-Abelian $S_3$ symmetry down to an Abelian $\mathbb{Z}_3$ symmetry. The relative strengths of these two terms are controlled by the $\kappa \in [0,1]$ parameter, where $\kappa=1$ is expected to be the SSB limit. 2) The second term $H_t$, the thermalizing term, contains a representative subset of the most general two-body symmetric operators. The intention is to keep $H_t$ sufficiently generic while retaining invariance under symmetry action (see Appendix~\ref{app:building symmetric ham} for details of how this Hamiltonian is constructed and can be generalized to arbitrary symmetry groups). The $\lambda$ parameter controls disorder strength. 

Using  the transformation of the spin operators listed above, it is straightforward to verify that the Hamiltonian has the desired symmetry, that is, $\forall g\in S_3$
	\begin{align}
	U(g) H(\lambda,\kappa) U(g)^\dagger = H(\lambda,\kappa),
	\quad U(g) = \bigotimes_{i=1}^L V(g)_i .
	\end{align}
For numerical analysis, the parameters in the Hamiltonian are selected as follows,
	(a) $h_i = w_h~g_h(i)$, $J_i=w_J~g_J(i)$ and $g_{h/J}(i)$ are random numbers drawn from a normal distribution with mean 0 and standard deviation 1;
	(b) $(w_h,w_J,\Delta_t,a,b,c)$ are free real parameters. We arbitrarily fix these to the values $(1.0,0.6,0.17,0.74,0.67,0.85)$ respectively for our numerical study without loss of generality.

%
 	
	\section{Numerical Results}~\label{sec:diagnostics}
	
	We perform exact diagonalization of the Hamiltonian in Eq.~\ref{eq:S3 Hamiltonian}. The local Hilbert space for the $S_3$ symmetric Hamiltonian is three dimensional, in contrast to the spin $\frac{1}{2}$ case and this constraints our maximum system size accessible to be $L=10$ sites. 	We study the properties of the eigenstates pertaining to localization and thermalization of this Hamiltonian for various values of $\lambda \in (0,\infty)$ and $ \kappa \in [0,1]$. For clarity of presentation of certain analysis, we use the rescaled variable $\frac{\lambda}{1+\lambda} \in (0,1)$ instead of $\lambda$ (wherever mentioned).  We employ periodic boundary conditions for all our analysis. To characterize the phases, we study below several relevant diagnostics that quantify the nature of localization and thermalization of the eigenstates. First, we consider the full entanglement distributions evaluated using the cut averaged entanglement entropy. 

	\subsection{Cut averaged entanglement entropy distributions}
	
	Since MBL is a characteristic of a single eigenstate, it is useful to quantify this phase  without averaging across different eigenstates. Recent work by Yu-Clark and Pekker~\cite{clark2016cutavg} has proposed cut averaged entanglement entropy (CAEE) to quantify the MBL phase at the level of a single eigenstate. CAEE is obtained by averaging the entanglement entropy for a given subsystem size across all possible locations of the subsystem on the spin chain. The CAEE scaling is then evaluated by repeating this procedure for different subsystem sizes. The key advantage of the CAEE is that strong sub-additivity condition constrains the shape of the entanglement scaling as a function of subsystem size, \thatis, CAEE is guaranteed to be a smooth convex function of the subsystem size without any average over disorder or eigenstates~\cite{clark2016cutavg}. This allows us to quantify the entanglement scaling of each eigenstate using the slope of the CAEE (SCAEE), as in Ref~\cite{clark2016cutavg}. We can then construct the full distribution of the slopes across the disorder snapshots and eigenstates. Fig.~\ref{fig:cut averaged plots} shows sample CAEE, along with a spline fit for 200 randomly chosen eigenstates for 4 different $\{\lambda, \kappa\}$ from a few disorder realizations of a 10-site Hamiltonian~[\ref{eq:S3 Hamiltonian}]. 
	\begin{figure}[!htb]
		\centering
		\begin{tabular}{cc}			
			\includegraphics[width=45mm]{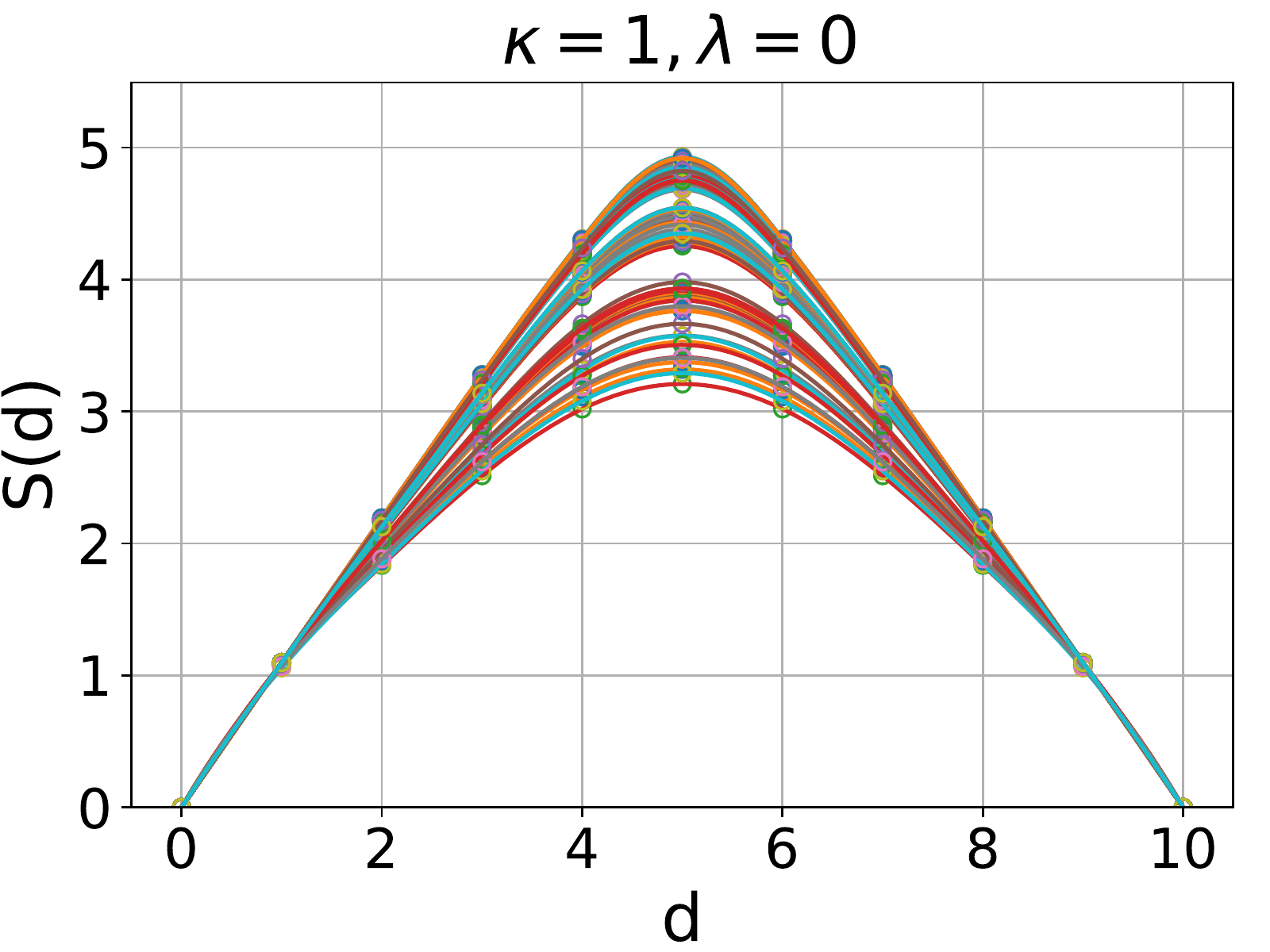}&			
			\includegraphics[width=45mm]{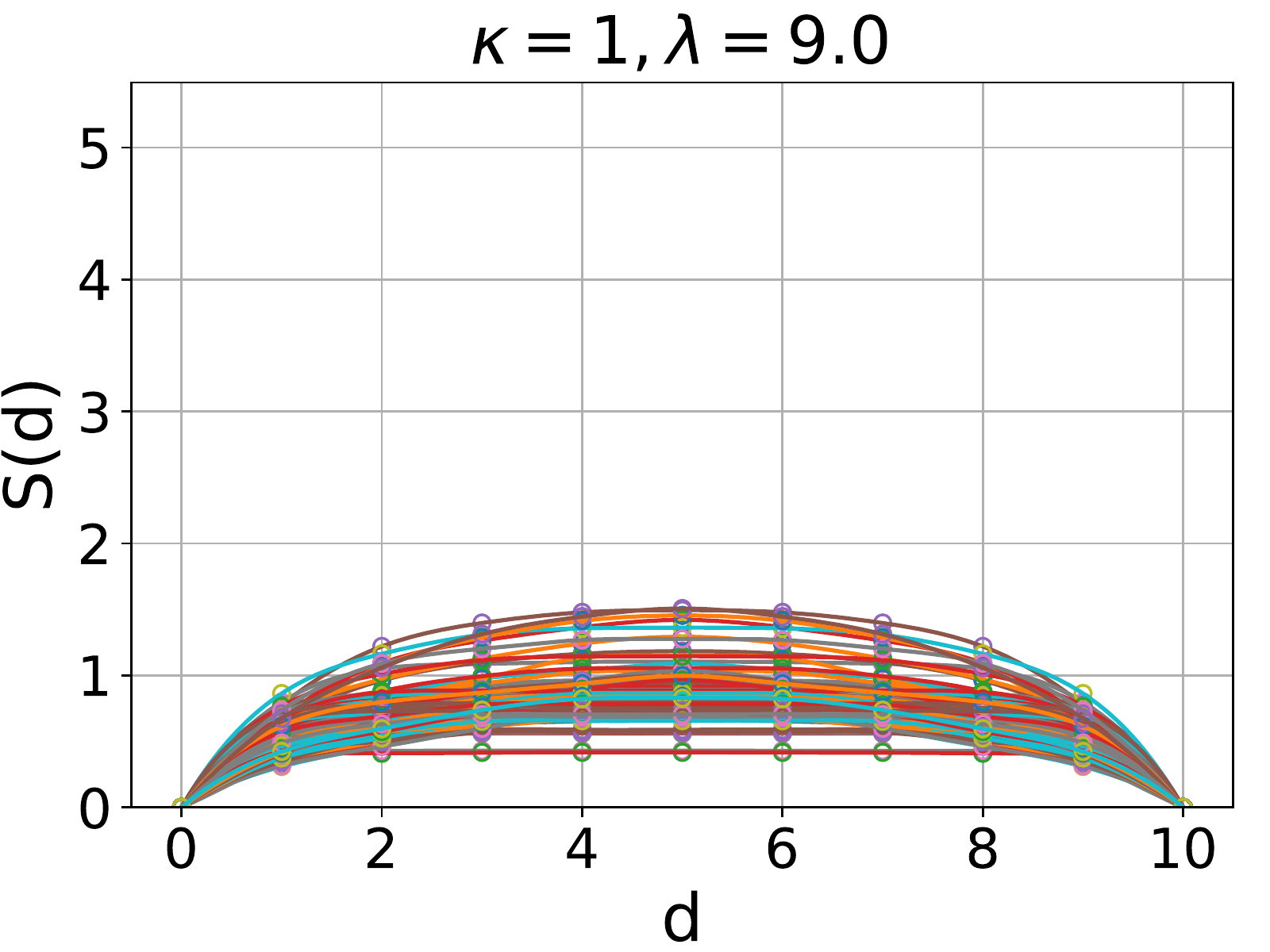}\\			
			\includegraphics[width=45mm]{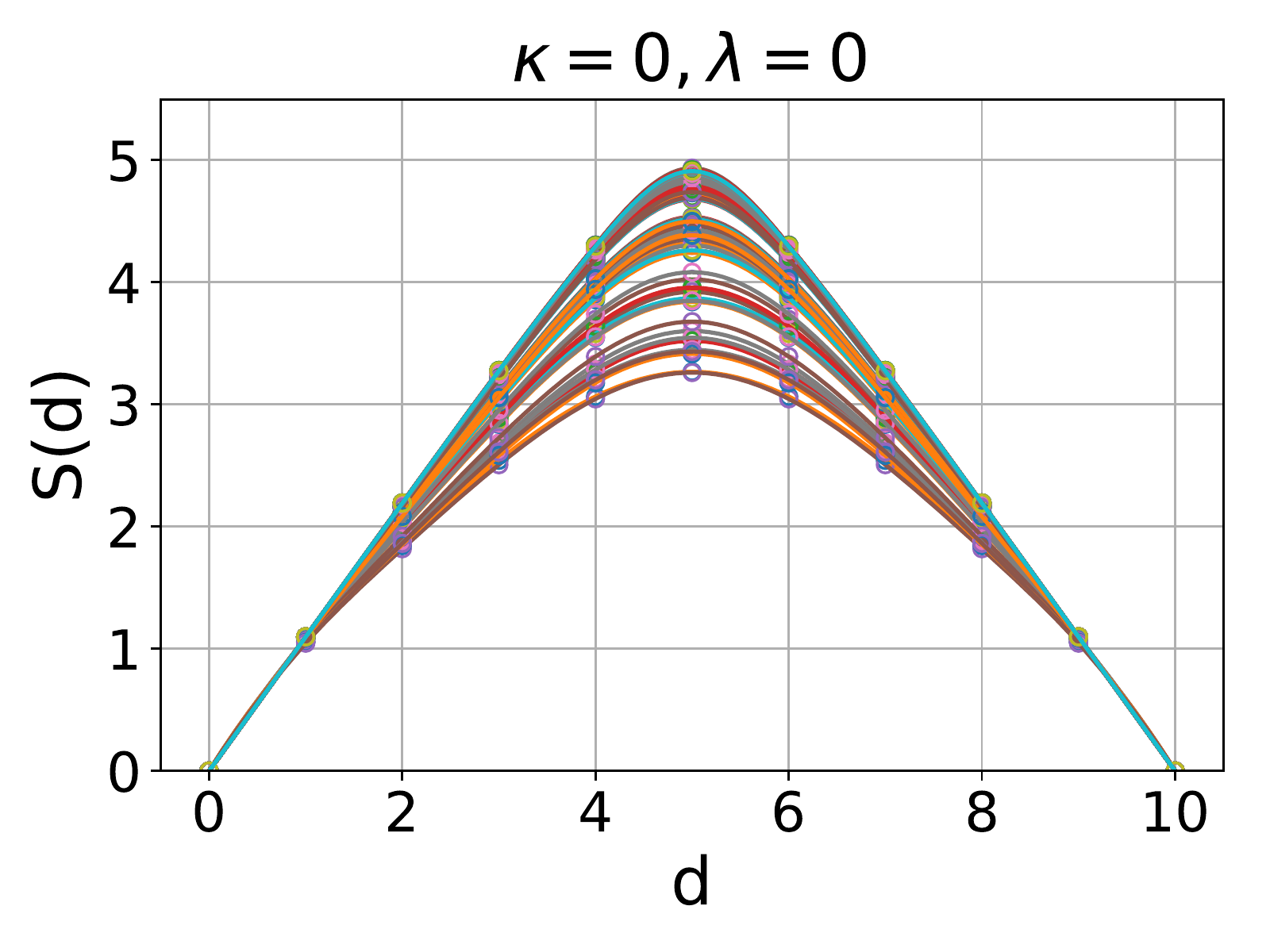}&			
			\includegraphics[width=45mm]{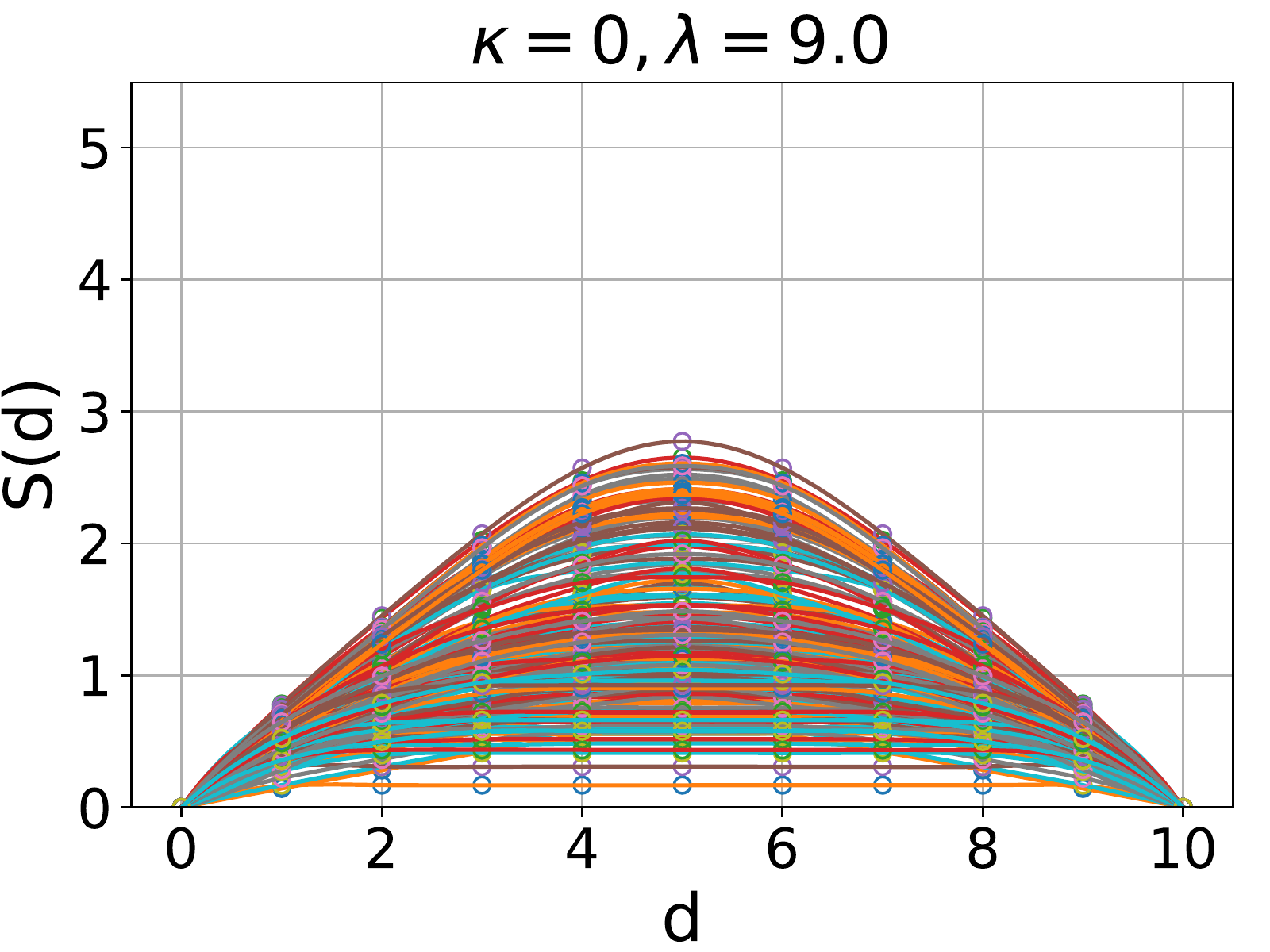}						
		\end{tabular}		
		\caption{(Color online) CAEE and spline fit for 200 eigenstates randomly sampled from the spectra of 19 disorder realizations of the 10 site Hamiltonian~(\ref{eq:S3 Hamiltonian}). \label{fig:cut averaged plots}}			
	\end{figure}
It can be seen that the eigenstates for small $\lambda$ are mostly volume law, while for large $\lambda$, there exists mixture of area-law and volume-law states. In order to identify the nature of the eigenstate transitions for various parameter regimes, we monitor the full distribution of the slope of the CAEE (SCAEE) evaluated at subsystem size $L/4$ for different values of $\lambda$ and $\kappa$ and disorder realizations. We compute SCAEE as follows: for each eigenstate, we compute the CAEE, $S(d)$ for different subsystem sizes, fit the data to a curve using a spline fit and then evaluate the slope, $S'(d)$ for this fit curve at $d=L/4$. 

There is however a potential issue because of the non-Abelian nature of the $S_3$ symmetry of the Hamiltonian. Generally, at finite-sizes, the eigenstates of a Hamiltonian invariant under the action of a symmetry group $G$ transform as irreps of the same group. For our case, $S_3$ has 3 irreps- two 1D irreps ($\mathbf{1}$, $\mathbf{1'}$) and one 2D irrep ($\mathbf{2}$). Eigenstates that transform as the $\mathbf{2}$ are two-fold degenerate. We may get different entanglement scaling depending on which precise orthonormal states in this 2D vector space is produced as the eigenstates~\footnote{We thank Bela Bauer for pointing this out}. To avoid this issue, we diagonalize the Hamiltonian in the 1D irrep sector. In addition to ensuring that we only sample non-degenerate eigenstates which transform as 1D irreps, this also helps us in reaching higher system sizes (See Appendix~[\ref{app:detecting irrep}] for more details). 

	\begin{figure}[!htb]
	 	 	\includegraphics[width=85mm]{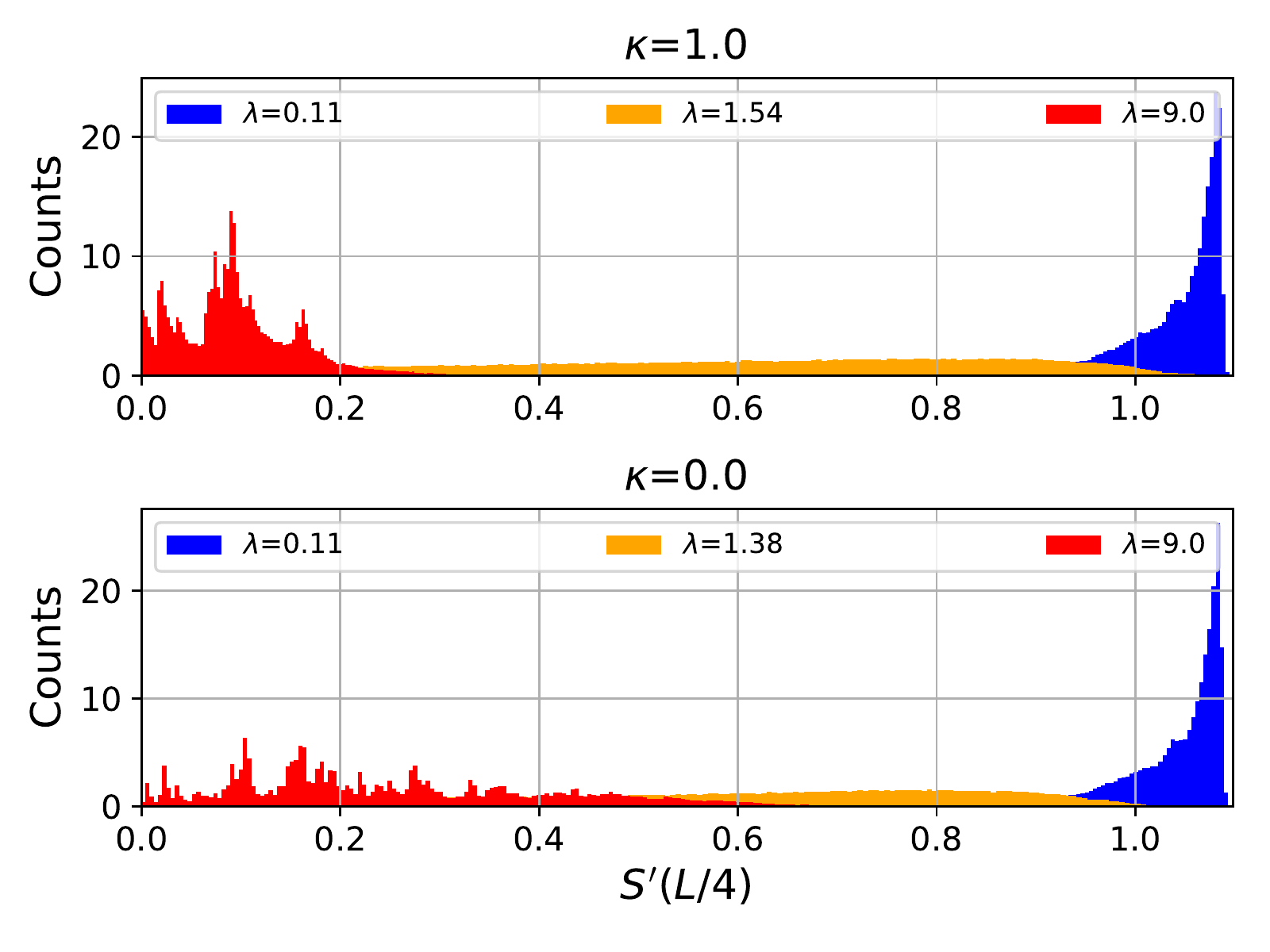}
	 	 	\caption{ (Color online) Slope histograms for 9 sites and 879 disorder samples for representative $\{\lambda,\kappa\}$. $243$ Eigenstates that transform as 1D irreps chosen sampled for each disorder realization. The plot is normalized to have unit area. \label{fig:slope histograms}}
	 	 \end{figure}	 	 

Fig.~\ref{fig:slope histograms} shows the distribution of $S'(L/4)$ for different $\lambda$ and $\kappa$. The values of $\lambda$ are chosen so as to show what the distribution looks like when the system has strong, weak and intermediate disorder strength. For all values of the SSB parameter $\kappa$, for weak disorder $\lambda$, we see that the $S'(L/4)$ distribution becomes increasingly narrow with system size with a peak located close to $1.1 \approx log(3)$ which is the maximum value possible for $S'(d)$ for any state of spin-1 chain. This is also evident from the large number of eigenstates with volume-law entanglement entropy scaling $S(d) \propto d$  in Fig.~\ref{fig:cut averaged plots}. These properties are consistent with an ergodic/ thermal phase. For high disorder however, we see that the distribution for $\kappa=0$ in Fig.~\ref{fig:slope histograms} is different from $\kappa=1$. For the former, there is a relatively extended thermal tail which is suppressed for the latter. To gain a better understanding we present the first two moments (mean and variance) of this slope distribution which are indicators of a potential MBL transition. 

{\it $\kappa=1$ with $S_3\rightarrow \mathbb{Z}_3$ symmetry breaking:} The entanglement distribution for the $\kappa=1$ limit Fig.~\ref{fig:slope histograms} and its moments, displayed in the upper panels of Figs.~\ref{fig:slope mean cuts},\ref{fig:slope var cuts} are consistent with the existence of an MBL phase for the large disorder limit and transitions to a fully thermal phase for weak disorder. One way to estimate the transition point is by locating where the mean, $\overline{S'(L/4)}$ curves for different system sizes cross on the $\lambda$ axis in Fig.~\ref{fig:slope mean cuts}. This is roughly at $\lambda/(1+\lambda) \approx 0.72$. Another is by locating the peak of the variance, $\sigma^2(S'(L/4))$ curve on the $\lambda$ axis which is believed to be close to the point of phase transition in the thermodynamic limit~\cite{husekhemani2016critical,husezhang2016floquet}. The drift of this point towards the disorder side \thatis larger $\lambda$ with increase in system size is considered to be typical for exact diagonalization (ED) studies of MBL~\cite{husekhemani2016critical}. Since our model has a non-Abelian symmetry, the existence of a full MBL phase must accompany SSB to an Abelian subgroup. We confirm SSB ($S_3$ to $\mathbb{Z}_3$ in this case) by computing a spin-glass diagnostic in \ref{sec:ea}. 

{\it $\kappa=0$ with full $S_3$ symmetry:} The entanglement distribution for the $\kappa=0$ limit, shown in Fig.~\ref{fig:slope histograms},  and its moments, displayed in the lower panels of Figs.~\ref{fig:slope mean cuts} and \ref{fig:slope var cuts} shows an enhanced variance and mean at the $\kappa=0$ for the large disorder limit. The enhanced mean value is an indication of the presence of sub-thermal volume law states and area law states. However, the crossing of the $\overline{S'(L/4)}$ curves persists (roughly at $\lambda/(1+\lambda) \approx 0.70$) as does the peak in the $\sigma^2(S'(L/4))$ plot. How this peak value changes as we approach the thermodynamic limit is an open question and we hope that better tools of numerical analysis like matrix product state methods can shed some light on this issue. We leave this for future work.

\begin{figure}[!htb]
	\centering
	\begin{tabular}{c}			
		\includegraphics[width=65mm]{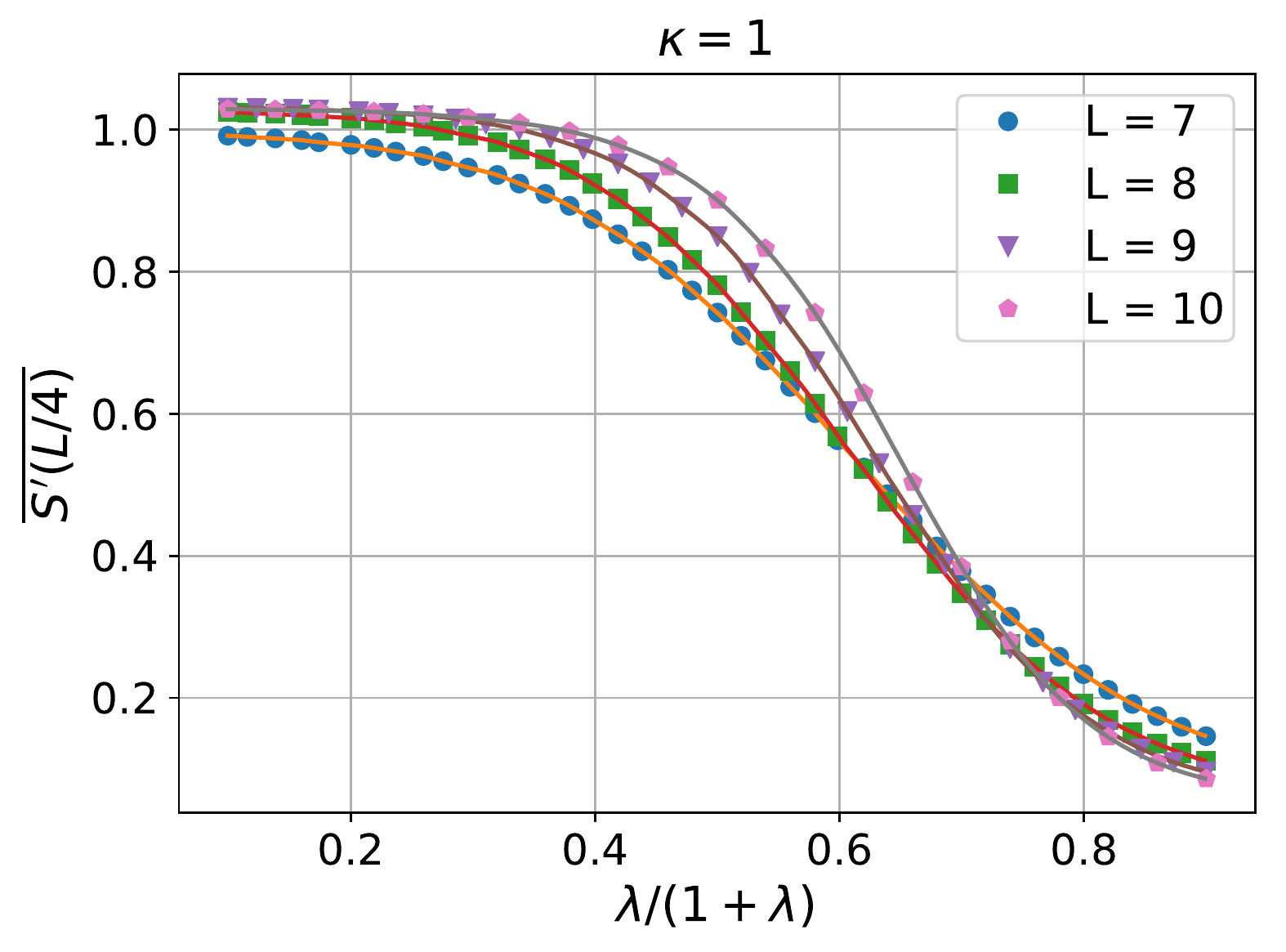}	\\
		\includegraphics[width=65mm]{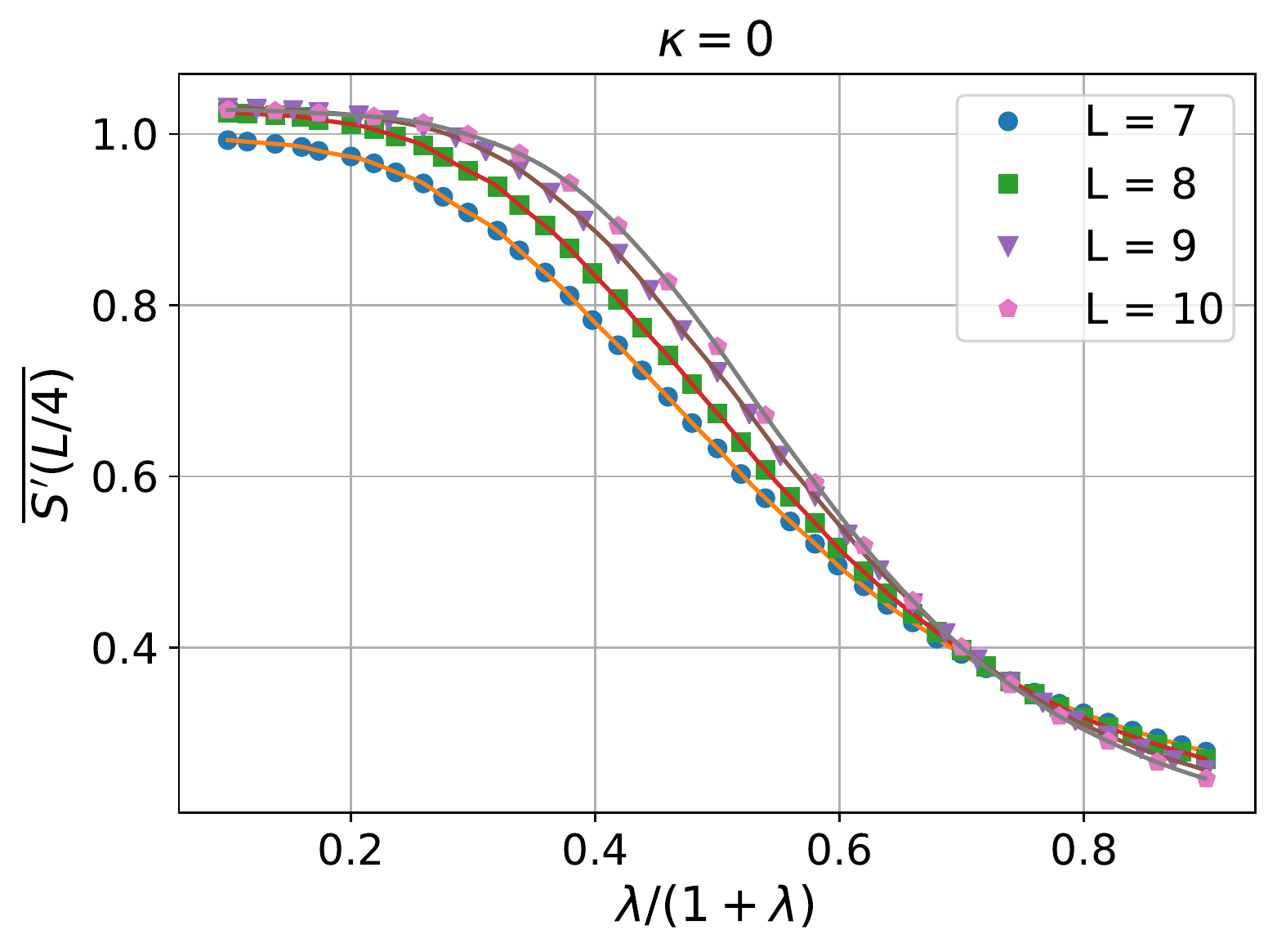}		 			
	\end{tabular}			 
	\caption{ (Color online) Mean of $S'(L/4)$ distribution (with spline fit) as a function of $\lambda/(1+\lambda)$ for $\kappa = 0$ and $\kappa = 1$. 243 eigenstates per disorder realization that transform as 1D irreps sampled for 800 (7,8 sites), 879 (9 sites) and  421 (10 sites) disorder realizations respectively .\label{fig:slope mean cuts}}
\end{figure}

	 \begin{figure}[!htb]
		\centering
		\begin{tabular}{c}			
	 	\includegraphics[width=65mm]{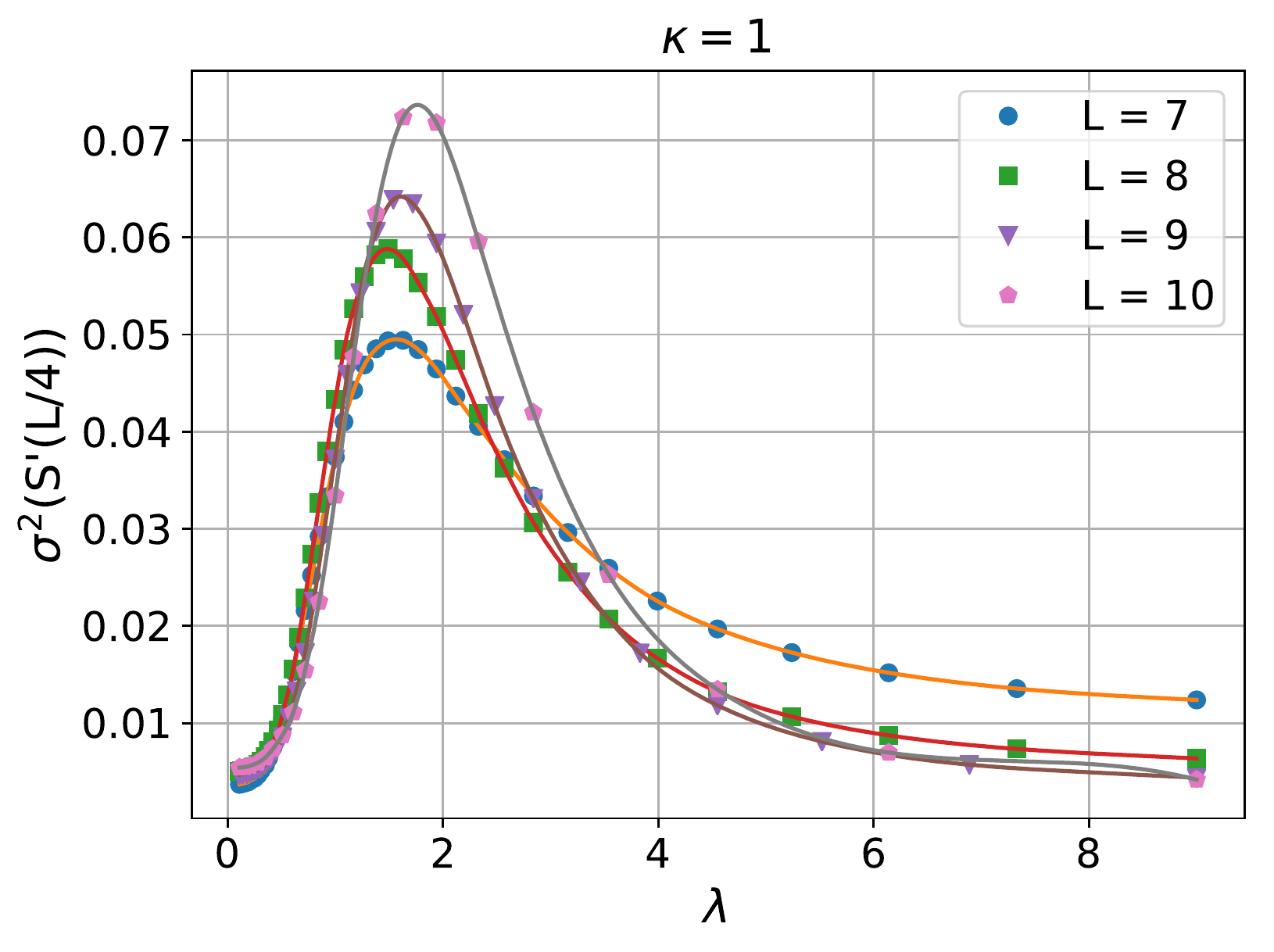}	\\
	 	\includegraphics[width=65mm]{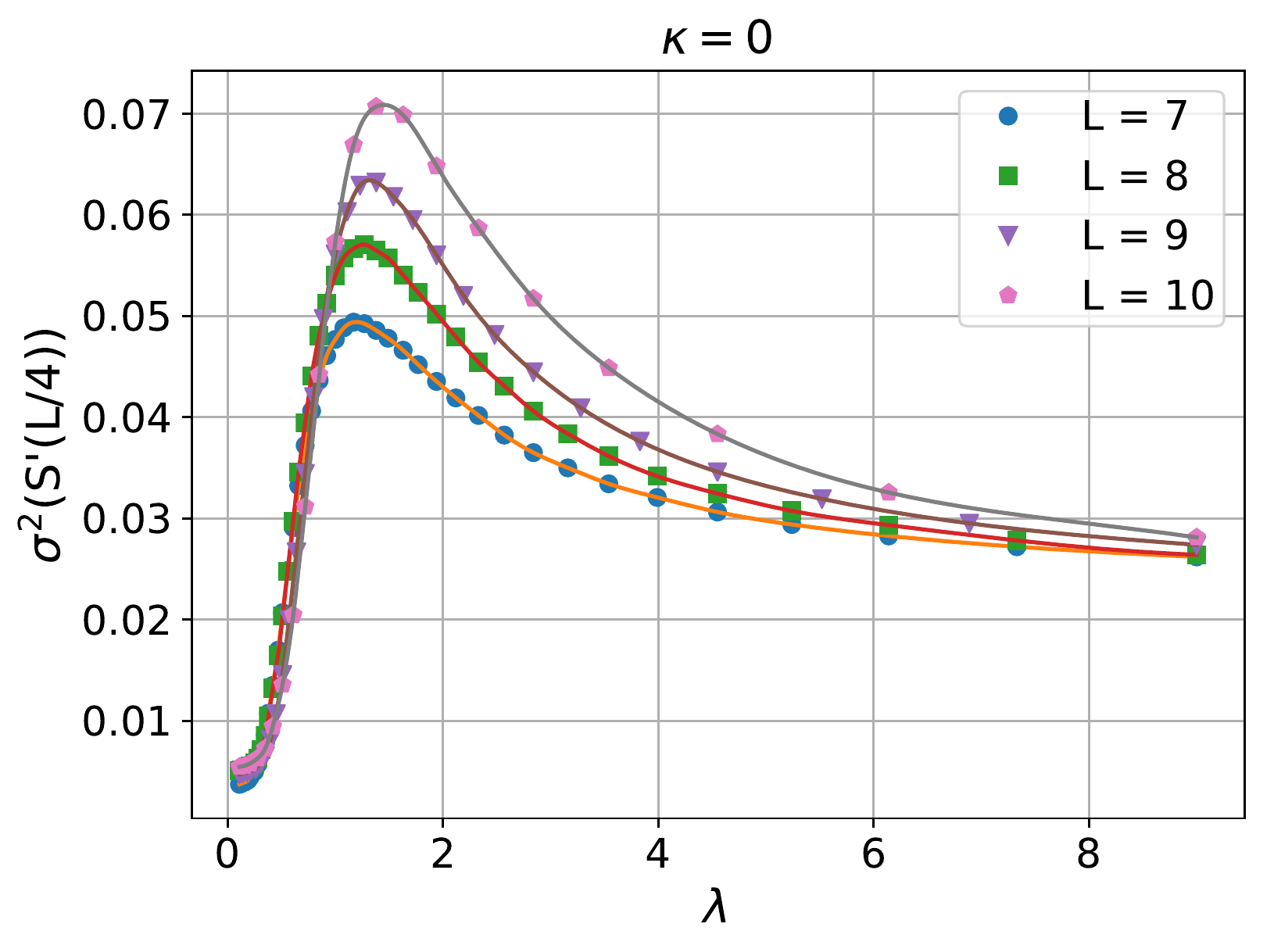}		 			
		\end{tabular}
		\caption{ (Color online) Variance of $S'(L/4)$ distribution (with spline fit) as a function of $\lambda$ for $\kappa = 0$ and $\kappa = 1$. 243 eigenstates per disorder realization that transform as 1D irreps sampled for 800 (7,8 sites), 879 (9 sites) and  421 (10 sites) disorder realizations respectively. \label{fig:slope var cuts}}		
	 \end{figure}	 
	 	 
\subsection{Spontaneous symmetry breaking in excited states} \label{sec:ea}
As indicated by the entanglement distributions, the full MBL phase appears only in the $\kappa=1$ limit. This is the limit where the disordered `l-bit' term is dominated by $\sum_i J_i~S^z_i S^z_{i+1}$ which triggers the spontaneous symmetry breaking (SSB). To confirm that SSB has indeed taken place for the many-body excited states, we use a spin-glass diagnostic which we describe below.  In the study of classical spin glasses (SG)~\cite{Edwards_Anderson,spinglass_binder}, one is interested in order parameters sensitive to the spin-glass phase characterized by long-range order in the presence of disorder. One such important quantity of study is the spin-glass susceptibility~\cite{spinglass_nishimori} 
	\begin{equation}
	\chi = \frac{1}{N} \sum_{i,j=1}^N \left[\langle s_i s_j \rangle^2\right],
	\end{equation}
where, $s_i$ are classical Ising variables, $\langle*\rangle$ indicates statistical averaging and $\left[*\right]$ indicates disorder averaging. In Ref~\onlinecite{bardarson2014}, the authors defined a similar quantum mechanical diagnostic to detect spin-glass order arising from SSB of $\ztwo \rightarrow~trivial~group$  in a $\ztwo$ invariant Ising- like disordered spin chain:
	\begin{equation} \label{eq:Pollmann_SG}
	\chi^{SG} = \frac{1}{L} \sum_{i,j = 1}^{L} |\innerproduct{\epsilon|\sigma^z_i \sigma^z_j}{\epsilon}|^2,
	\end{equation}
where, $\ket{\epsilon}$ is an energy eigenstate and $\sigma^z$ is the Pauli-Z operator. In their model, it was shown that the average $\overline{\chi^{SG}}$  scales with system size as $\overline{\chi^{SG}}\sim L$ in the SG phase and approaches a constant value set by normalization for the paramagnetic phase. Similar to eq.~\ref{eq:Pollmann_SG}, for our model, we define the following spin-glass diagnostic that looks for signatures of SG order arising from SSB of $S_3 \rightarrow \zthree$
	\begin{equation} \label{eq:SG_Z3}
		\ea=\frac{1}{L-1} \sum_{ i \neq j =1}^{L}|\innerproduct{\epsilon|S^z_i S^z_j}{\epsilon}|^2.
	\end{equation}
Note that we choose to exclude the $i=j$ term in the summation unlike Eq~[\ref{eq:Pollmann_SG}]. We look at the statistics of $\ea$ across randomly sampled eigenstates and disorder realizations. We find signatures for transition to an SG phase as we vary $\lambda$ similar to what was found in Ref~\onlinecite{bardarson2014}. Fig.~\ref{fig:ea mean cuts} shows $\overline{\ea}$	versus  $\lambda$ for $\kappa=0,1$. 
 \begin{figure}[!htb]
	\centering
 	\begin{tabular}{c}			
 		\includegraphics[width=65mm]{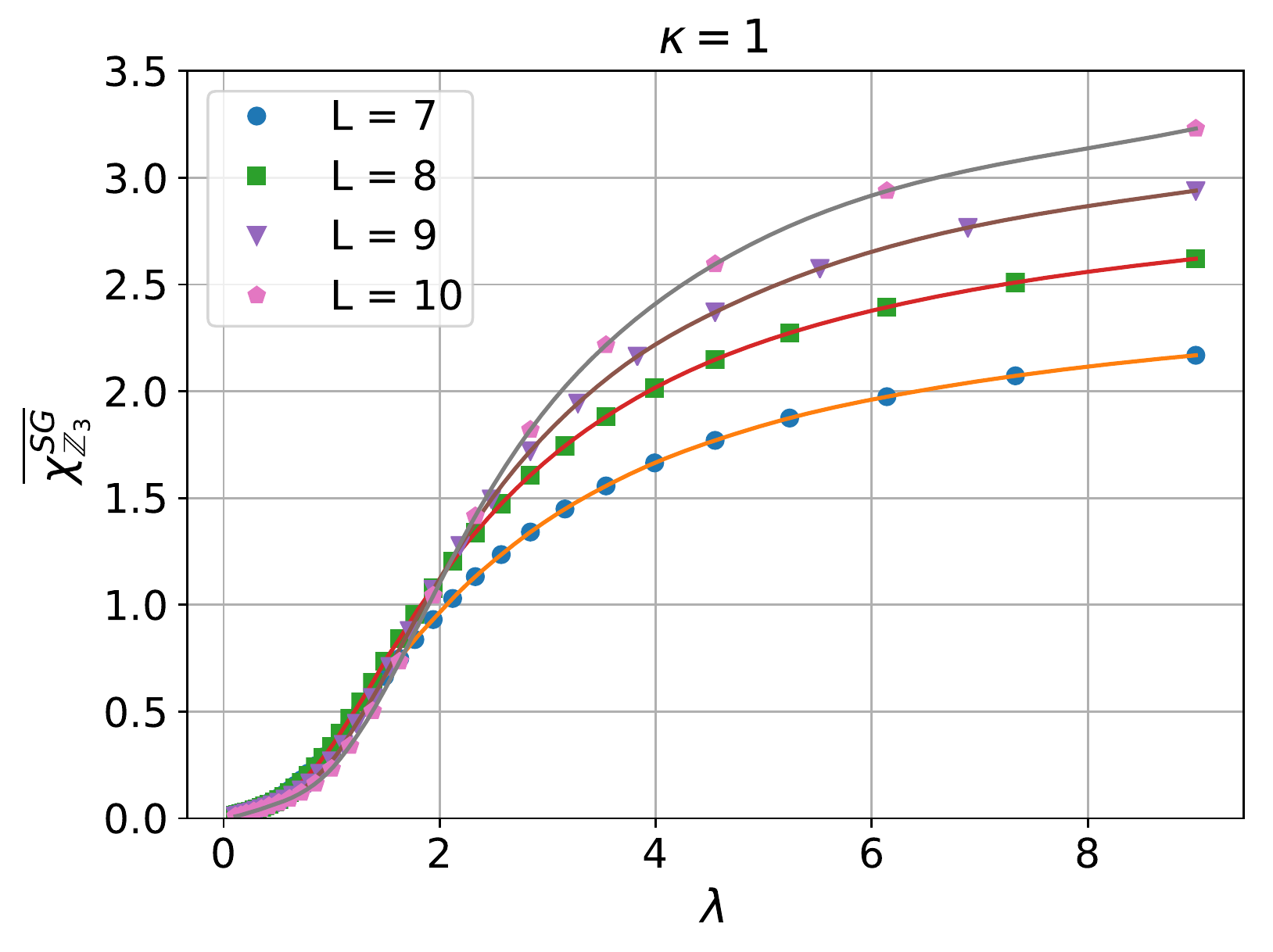}	\\
 		\includegraphics[width=65mm]{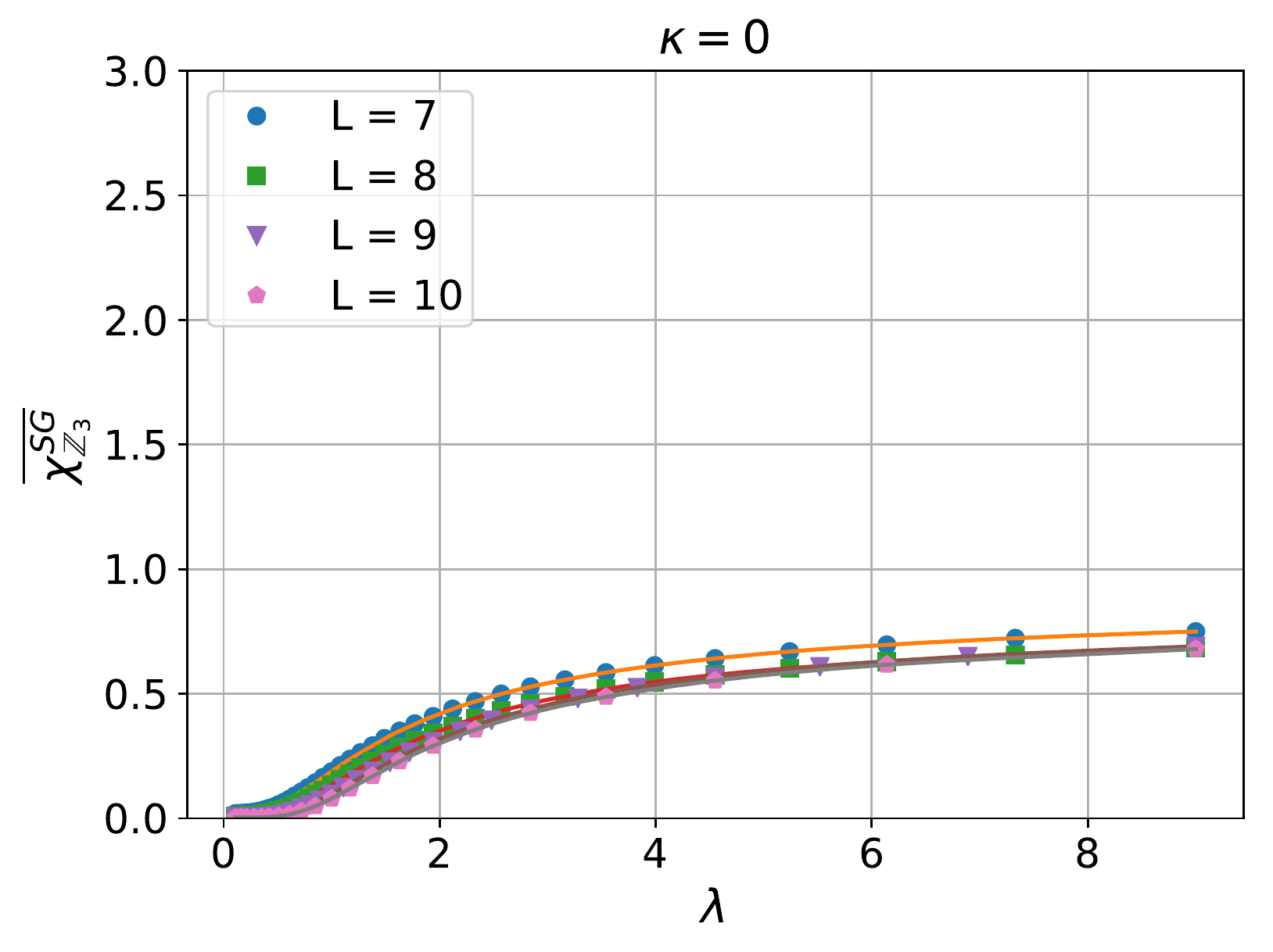}		 			
 	\end{tabular}
	\caption{ (Color online) $\overline{\ea}$ versus $\lambda$ (with spline fit) for $\kappa = 0$  and $\kappa= 1$. 243 eigenstates per disorder realization that transform as 1D irreps sampled for 800 (7,8 sites), 879 (9 sites) and  715 (10 sites) disorder realizations respectively. \label{fig:ea mean cuts}}	 					
 \end{figure}	 		
For $\kappa=1$, we indeed observe that the SG diagnostic increases with system size for large disorder. For $\kappa=0$, we see that the SG diagnostic at best saturates to a constant value independent of system size and at worst reduces with system size but certainly does not increase with it indicating the lack of SSB. To rule out SSB to other subgroups  in the $\kappa=0$ 	regime, we have to use other SG diagnostics that can detect spontaneous breaking of $S_3$  down to one of the other subgroups like $\ztwo$ and trivial group. In Appendix~\ref{app:SG}, we construct appropriate SG diagnostics and present numerical evidence that the full $S_3$ is indeed intact for $\kappa=0$. 

\subsection{Level statistics} \label{sec:levelstats}
{ Level statistics is a basis independent diagnostic that indicates localization and thermalization based on the statistics of the adjacent gap ratio ($\delta_n=E_{n+1}-E_{n}$) defined as,
\begin{align}
r_n=\min(\delta_n,\delta_{n+1})/	\max(\delta_n,\delta_{n+1}).
\end{align}

\begin{figure}[!htb]
	\centering
	\begin{tabular}{c}			
		\includegraphics[width=65mm]{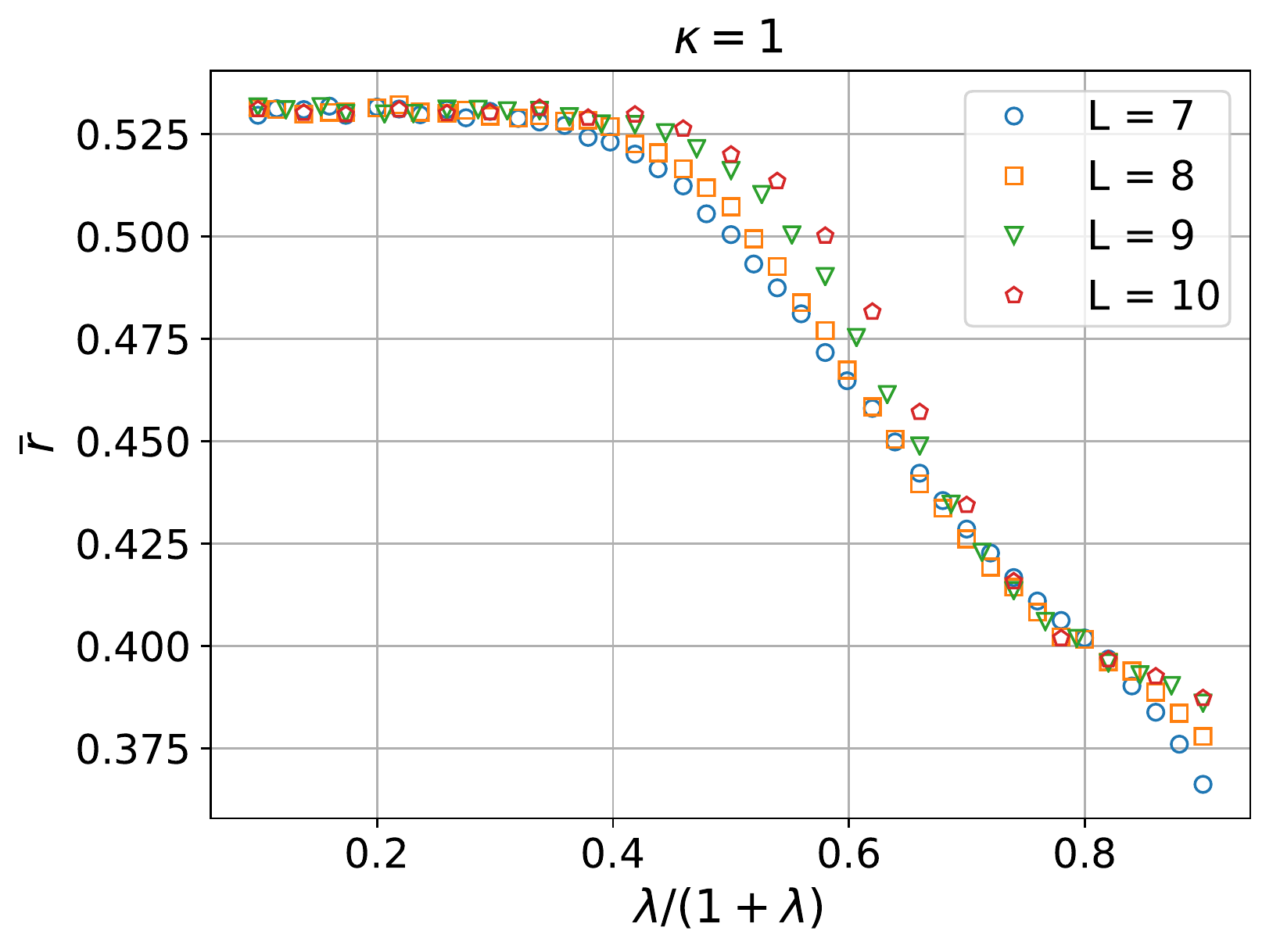}	\\
		\includegraphics[width=65mm]{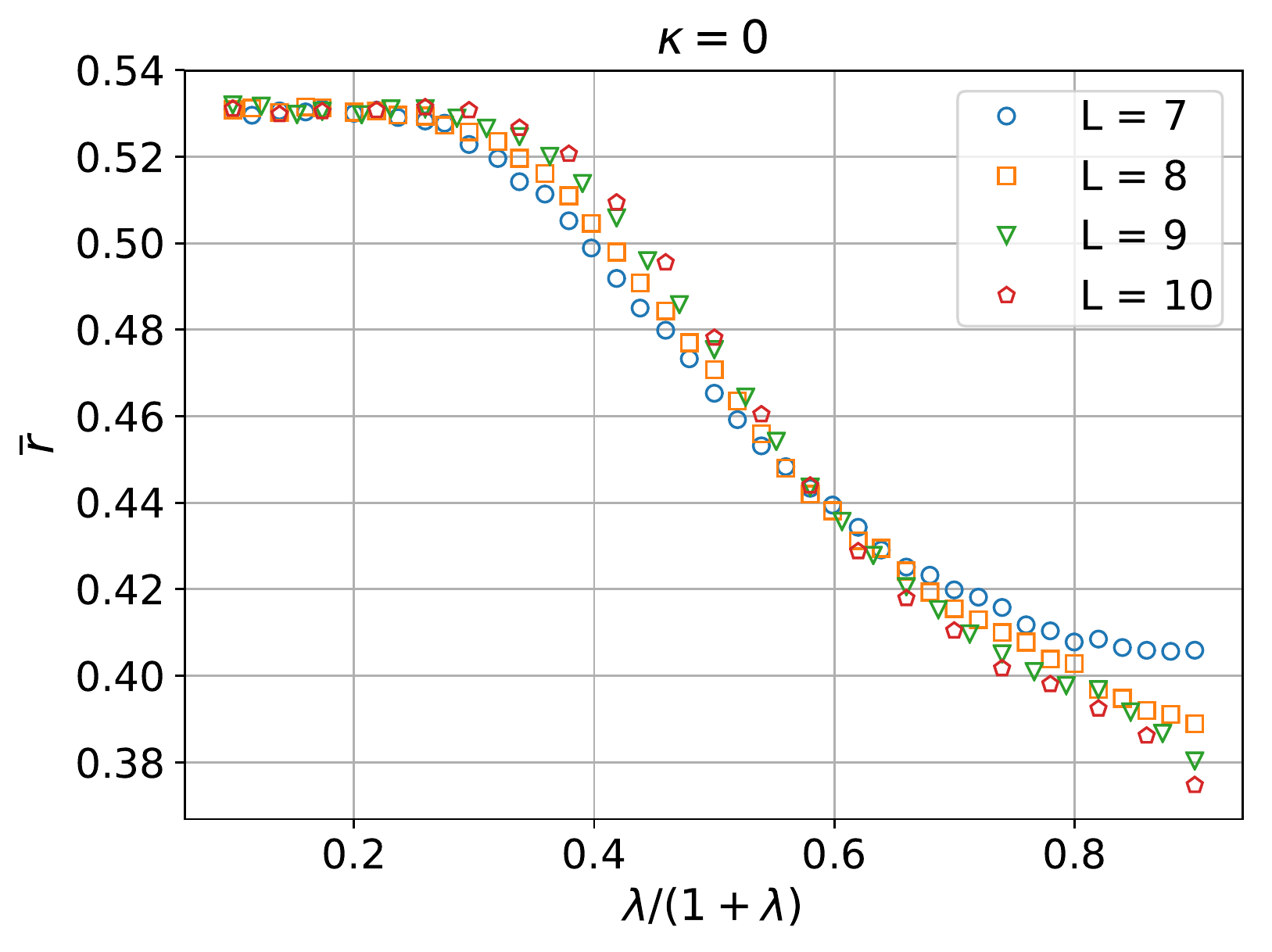}		 			
	\end{tabular}		
	\caption{ (Color online) $\overline{r}$ versus $\lambda /(1+\lambda)$ for $\kappa = 0$  and $\kappa= 1$ .  243 eigenstates per disorder realization that transform as 1D irreps sampled for 800 (7,8 sites), 879 (9 sites) and  715 (10 sites) disorder realizations respectively.  \label{fig:level mean cuts}}	 					
\end{figure}	

In the presence of symmetry, no interaction term in the Hamiltonian can mix eigenstates that transform as different irreducible representations of the symmetry group. Thus, it is meaningful to compute $r_n(\Gamma)$, the level statistics ratio for the eigenstates that transform as each irrep $\Gamma$ of the group separately. In this paper, in Fig.~\ref{fig:level mean cuts}, we present $\bar r = \frac{\overline{r}(\mathbf{1})+\overline{r}({\mathbf{1'}})}{2}$ as a function of $\lambda$ where $\overline{r}(\Gamma)$ is obtained by averaging $r_n(\Gamma)$ computed for randomly sampled eigenstates which transform as the $\Gamma$ irrep (see Appendix~\ref{app:detecting irrep} for details on how we detect the irrep of each eigenstate), across disorder realizations. On the side with low disorder, for both $\kappa=0$ and $\kappa=1$, we observe $\bar r\approx 0.53$. This is typical for a fully thermal phase indicating a Wigner-Dyson (WD) distribution. For a fully localized phase, typically, $\bar r\approx 0.39$, corresponding to Poisson distribution. However, the application of this diagnostic to the  degenerate spectrum of our $S_3$ invariant Hamiltonian can be tricky. In particular, closer to the l-bit point of large $\lambda$, the extensively degenerate spectrum may mimic level clustering and result in $\bar r< 0.39$ which we do indeed observe.  Hence, this diagnostic is only reliable at the thermal side, where we obtain the WD value of $\bar r\approx 0.53$. For the strong disorder, the level statistics approaches Poisson value but a clear reading is plagued by the degeneracies for both $\kappa=0,1$. }

\subsection{Finite size scaling}
In order to determine the location and nature of the putative transitions, we perform finite size scaling collapse for MBL and SG diagnosics using the following scaling ansatz used in Ref~\cite{bardarson2014}.
\begin{equation}
g(\lambda,L) = L^a f((\lambda-\lambda_c)L^{\frac{1}{\nu}})
\end{equation} 
Fig.~\ref{fig:scaling} shows the scaling collapse of the diagnostics $\overline{S'(L/4)},~\sigma^2(S'(L/4))$ and $ \overline{\ea}$ for 8, 9 and 10 site data. For $\kappa=1$, we obtain a good collapse for both the MBL and SG diagnostics at $\lambda_c \approx 2.65$ and $\nu \approx 2.5$. For $\kappa=0$, we do not get a scaling collapse for the SG diagnostic $\ea$ for non-zero $\lambda_c$ and positive $\nu$ which is consistent with the absence of SSB. On the other hand, we get a reasonable collapse for $\overline{S'(L/4)}$ and $\sigma^2(S'(L/4))$ at $\lambda_c  \approx 2.35$ and $\nu \approx 2.5$. 

Note that the value of the finite size exponent $\nu$ used in Fig.~\ref{fig:scaling} is consistent with the Harris~\cite{Harris}/CCFS~\cite{CCFS1,CCFS2}/CLO~\cite{CLO} criterion of $\nu \ge 2$ for one dimensional spin chains with quenched disorder. However, we emphasize that the our estimate is very rough and even for values of $\nu$ that violates the criterion (eg: $\nu \approx 1.5$ ), we still get a decent collapse. This violation is a common feature of ED studies of small system sizes~\cite{Louitz_nuviolation,bardarson2014}. Indeed, recent work by Khemani and Huse~\cite{khemaniTwoUniversality} suggests that this might be an indication of the system not exhibiting true thermodynamic behavior and is expected to undergo a crossover after which we obtain $\nu$ consistent with the Harris/CCFS/CLO bounds. Within the accuracy of our numerical investigation however, we can neither confirm nor rule out the possibility of our system being en-route to such a crossover. 

\begin{figure*}[!htb]
	\centering
	\begin{tabular}{ccc}			
		\includegraphics[width=60mm]{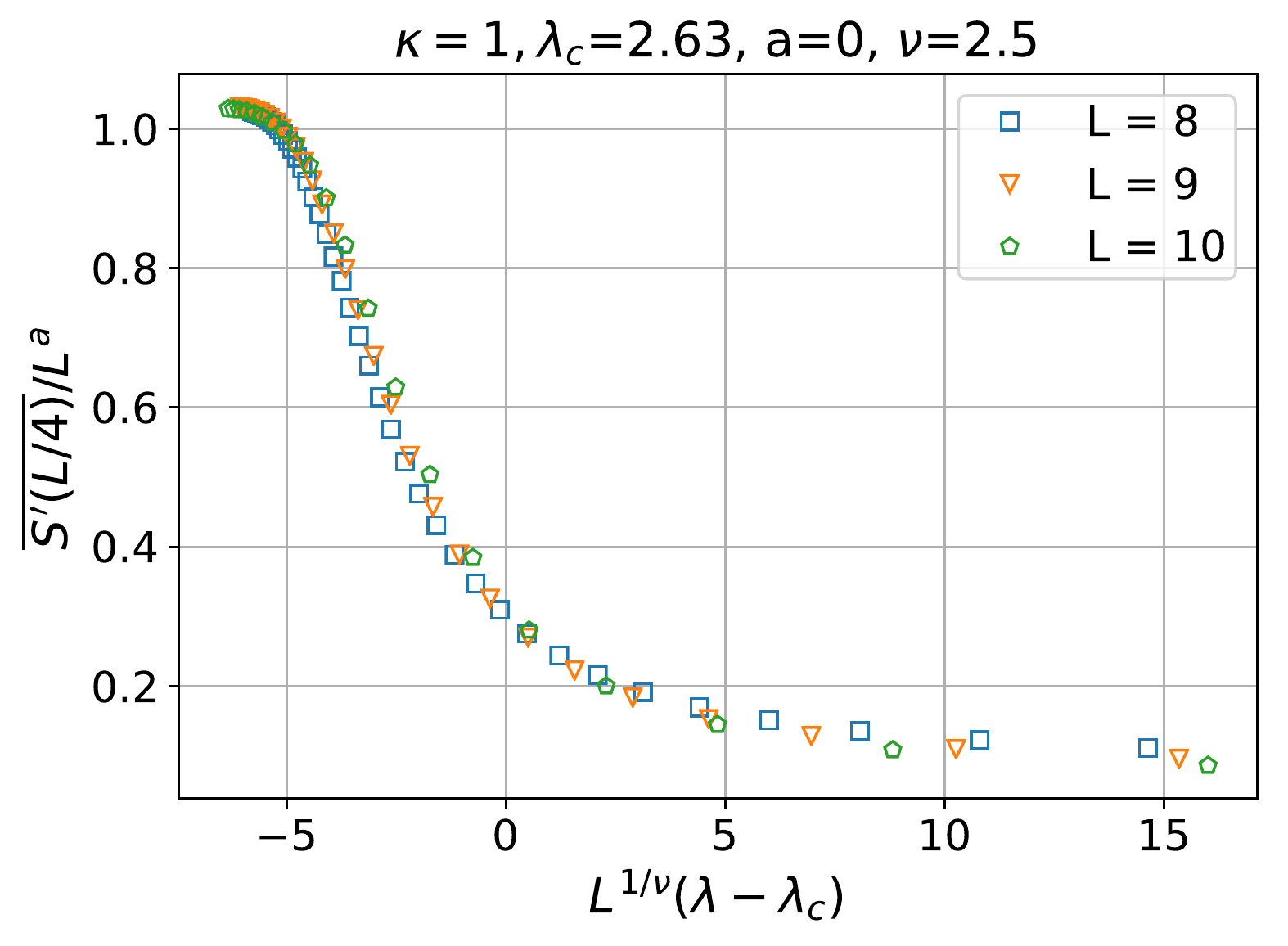}	&
		\includegraphics[width=60mm]{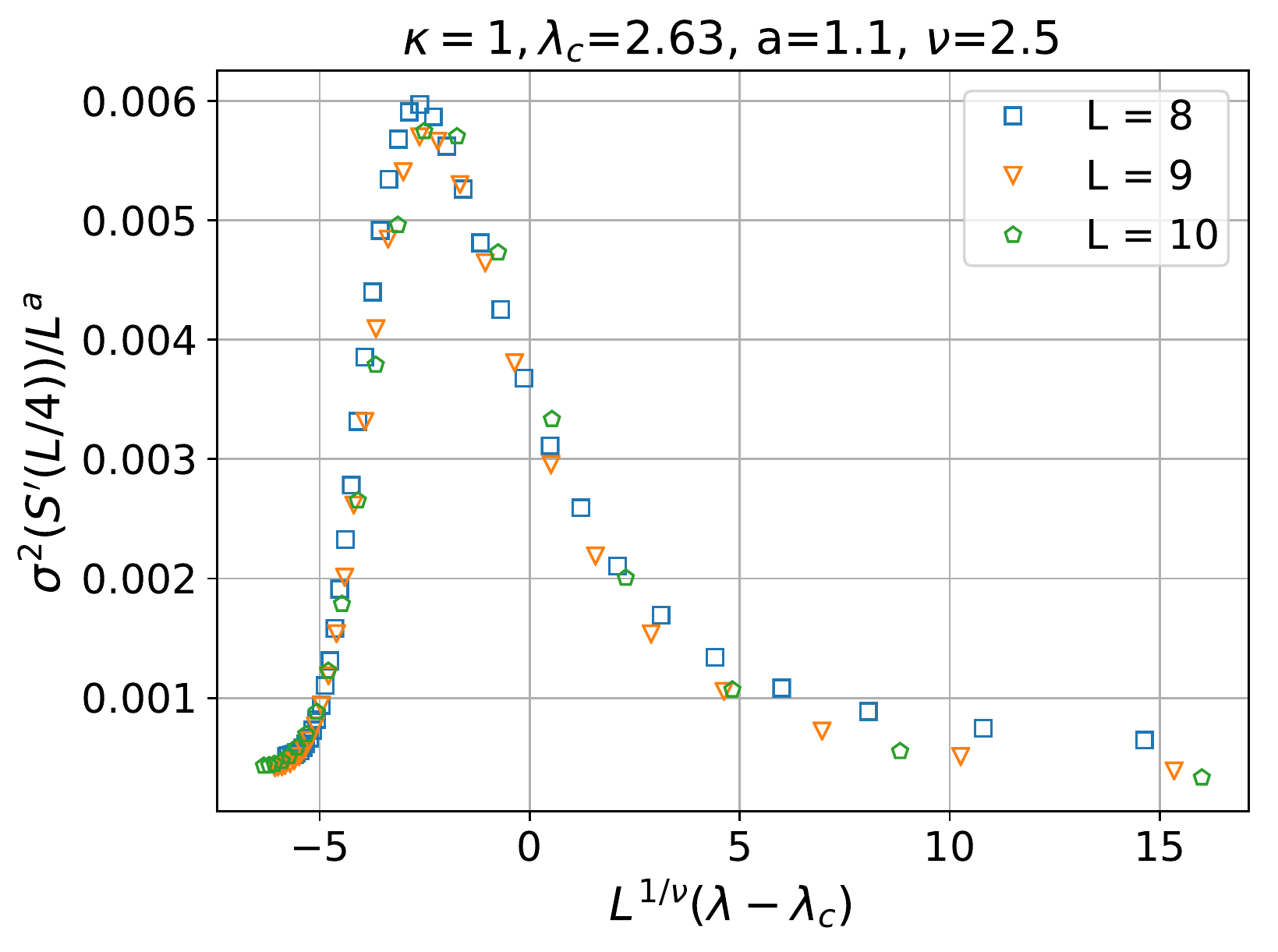}	&		
		\includegraphics[width=60mm]{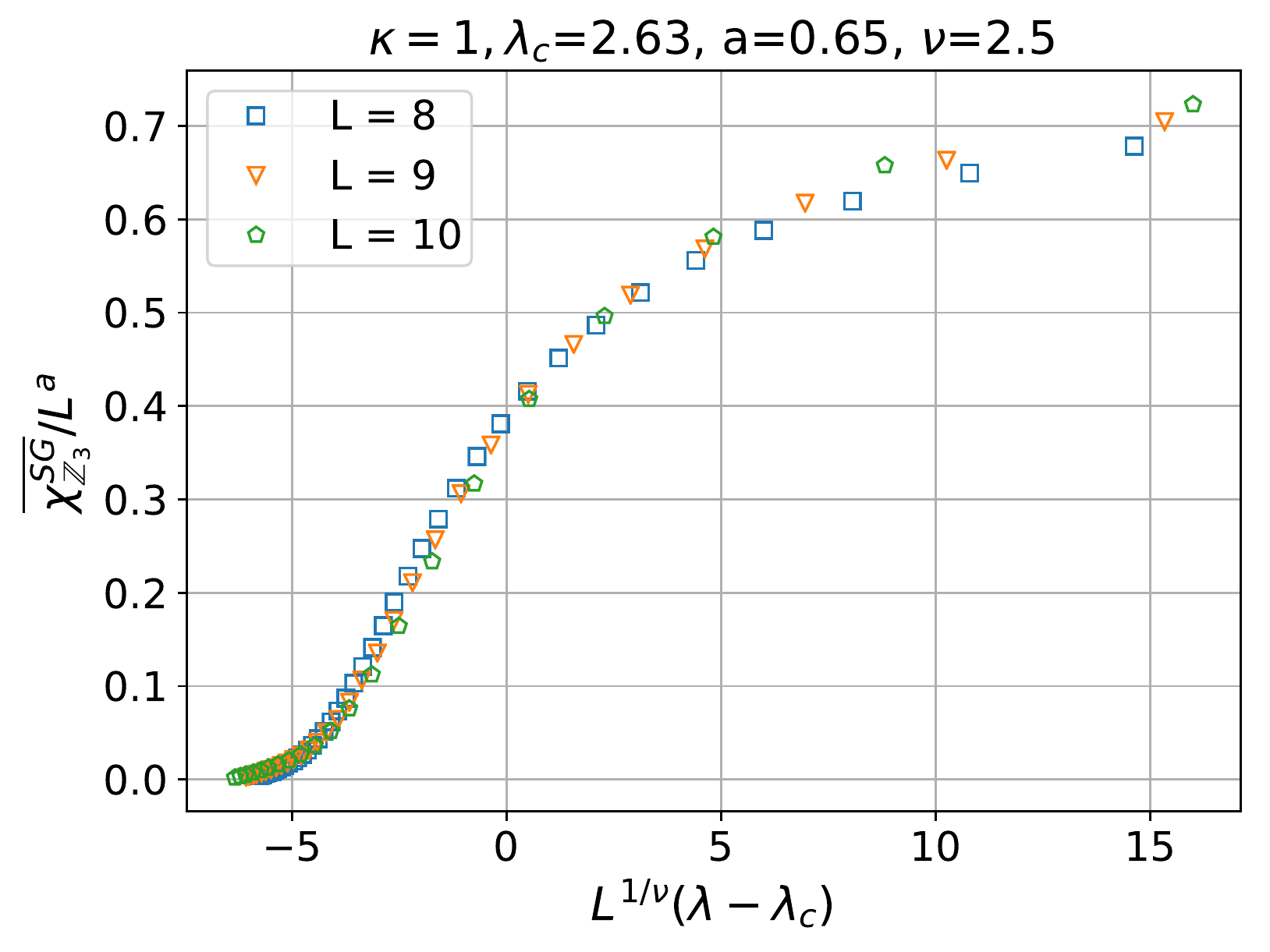}	\\	 						 					
		\includegraphics[width=60mm]{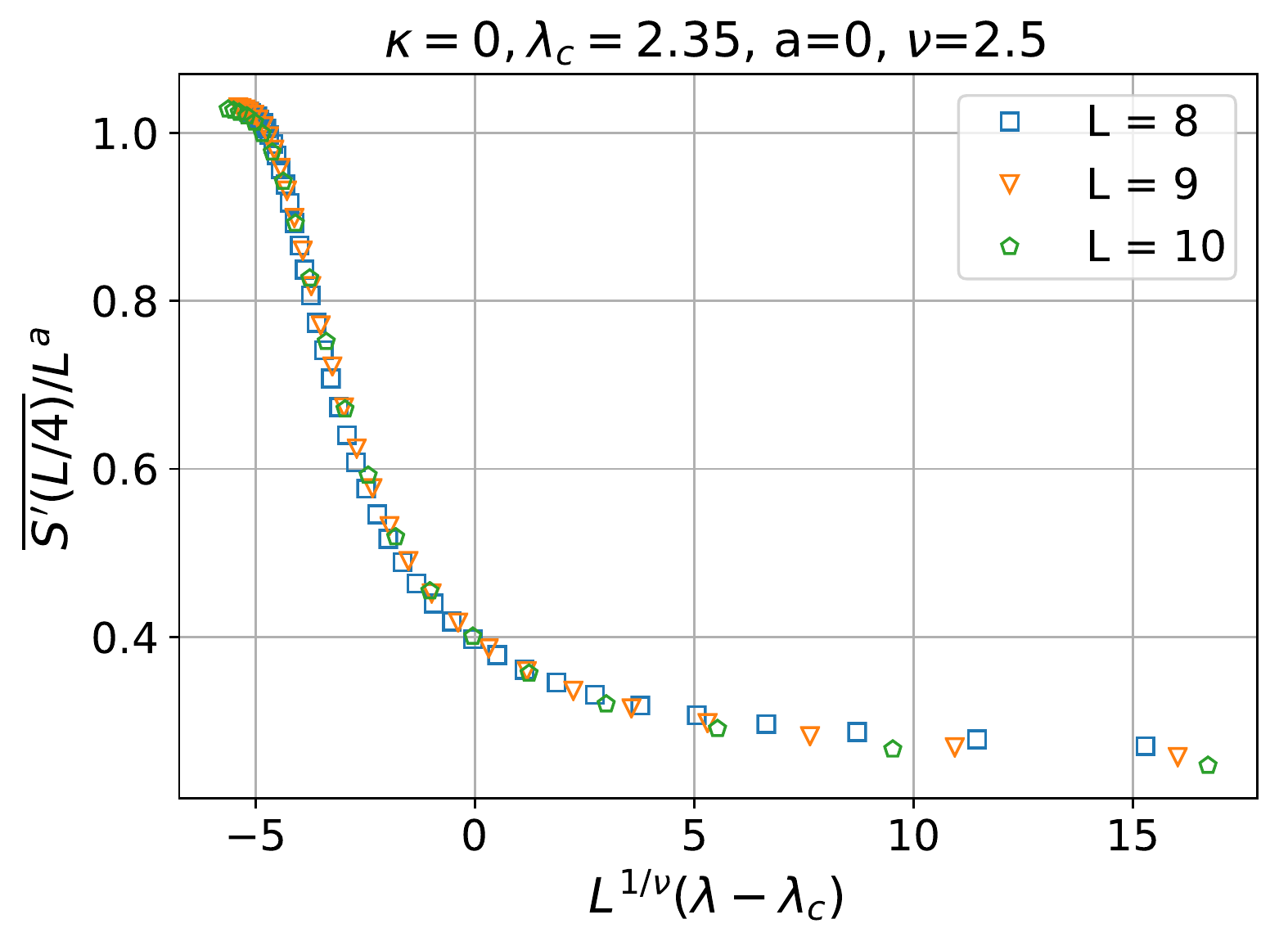}	&
		\includegraphics[width=60mm]{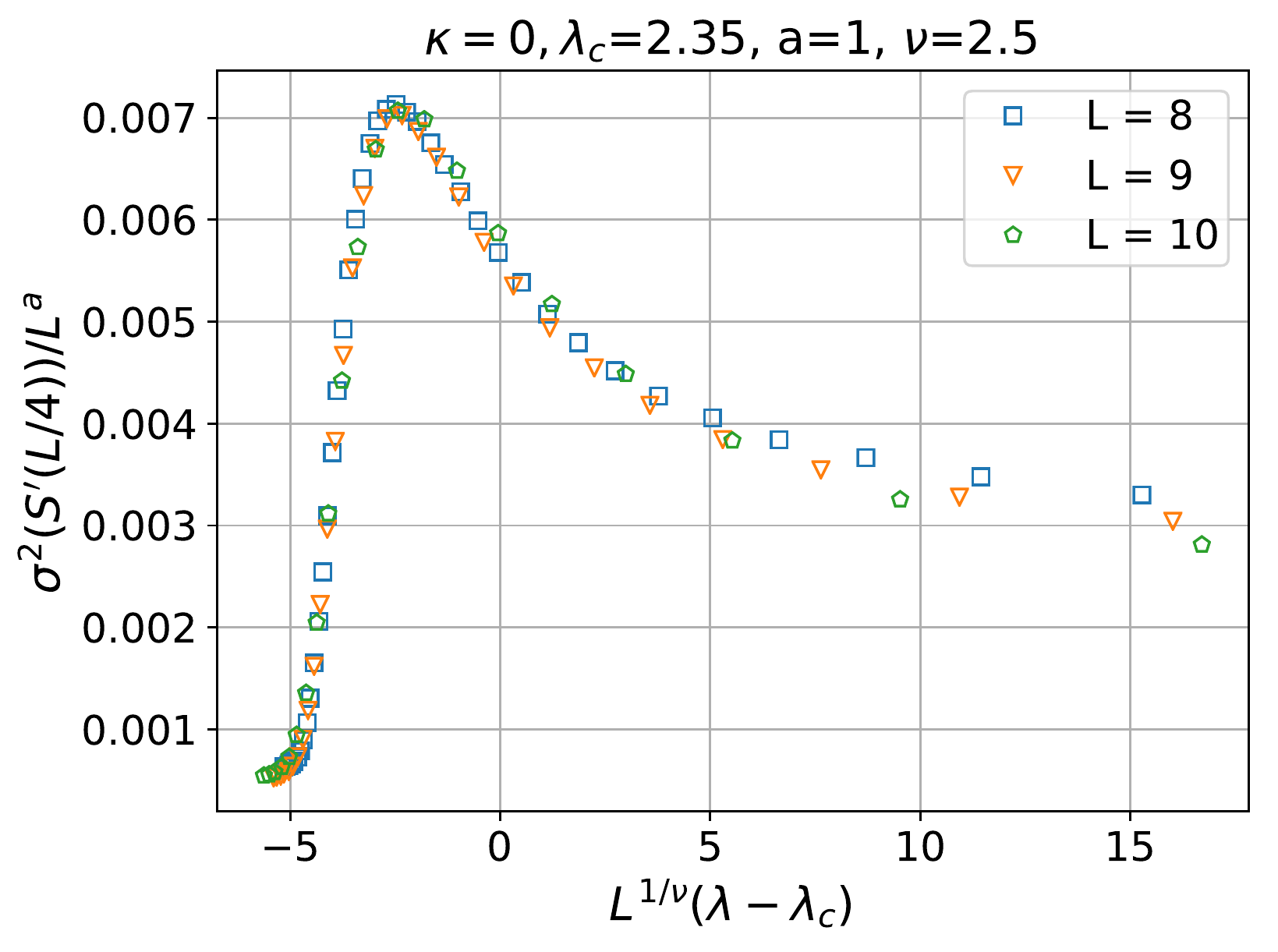}			
	\end{tabular}		
	\caption{(Color online) Scaling collapse of MBL and SG diagnostics for 8, 9 and 10 sites.}
	\label{fig:scaling}
\end{figure*}

\section{Discussion}
Based on the above analysis, we are in a position to put together an approximate phase diagram for the model Hamiltonian $H$. Fig.~\ref{fig:phase digs} shows a color map of $\overline{S'(L/4)}$, $\overline{\ea}$ and $\overline{r}$ plotted in the $\kappa, \lambda$ space. Fig.~\ref{fig:labeled phase diagram} shows a schematic plot that indicates a thermal phase for the weak disorder limit and a MBL+SSB phase for strong disorder limit. { In addition to these two phases, there is a hint of a third regime for strong disorder and no SSB, where there seems to be a coexistence of localized and delocalized states in the many-body spectrum. We now discuss each of this regimes separately.
\begin{figure*}[!htb]
	\centering
	\begin{tabular}{ccc}			
		\includegraphics[width=60mm]{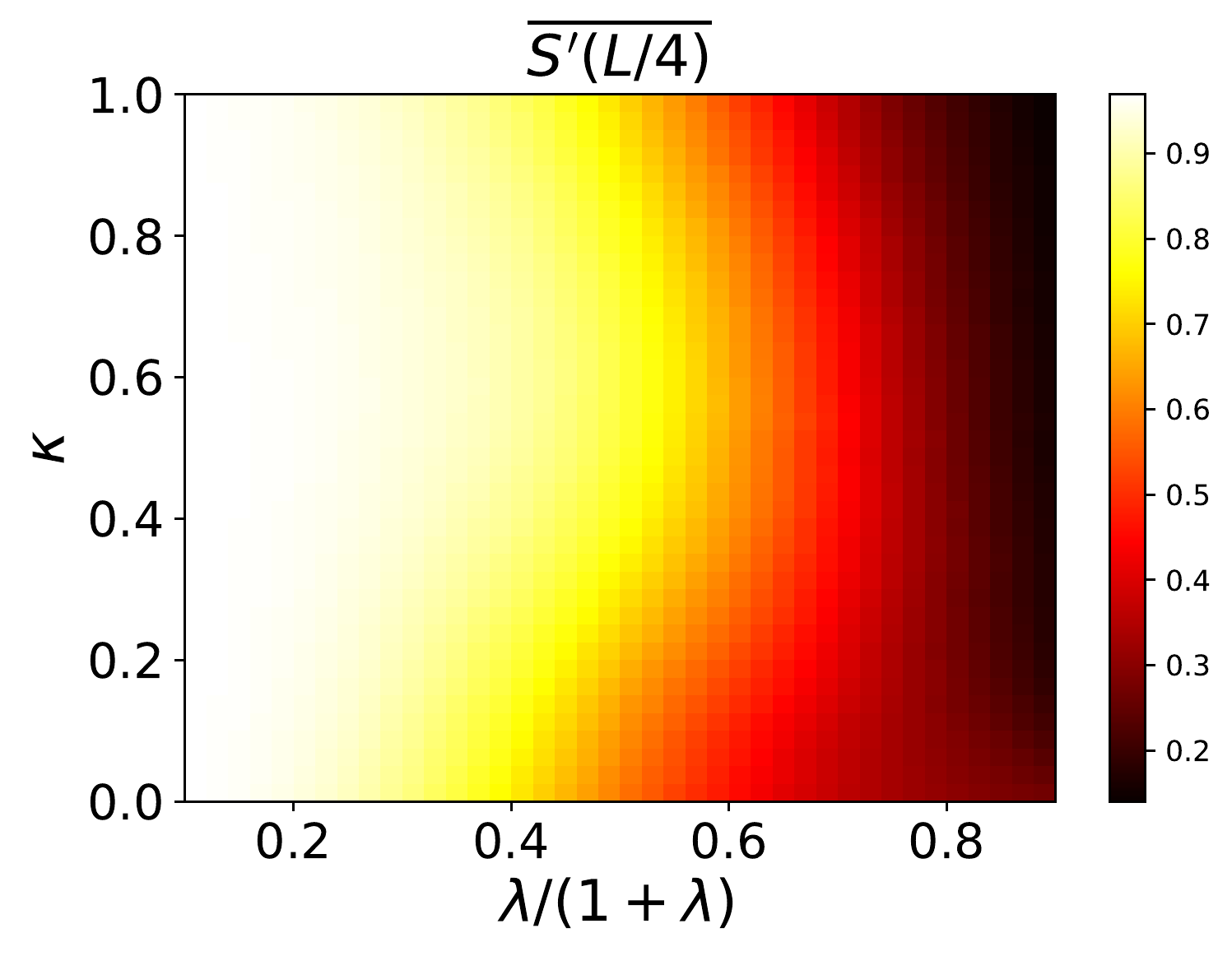}	&
		\includegraphics[width=60mm]{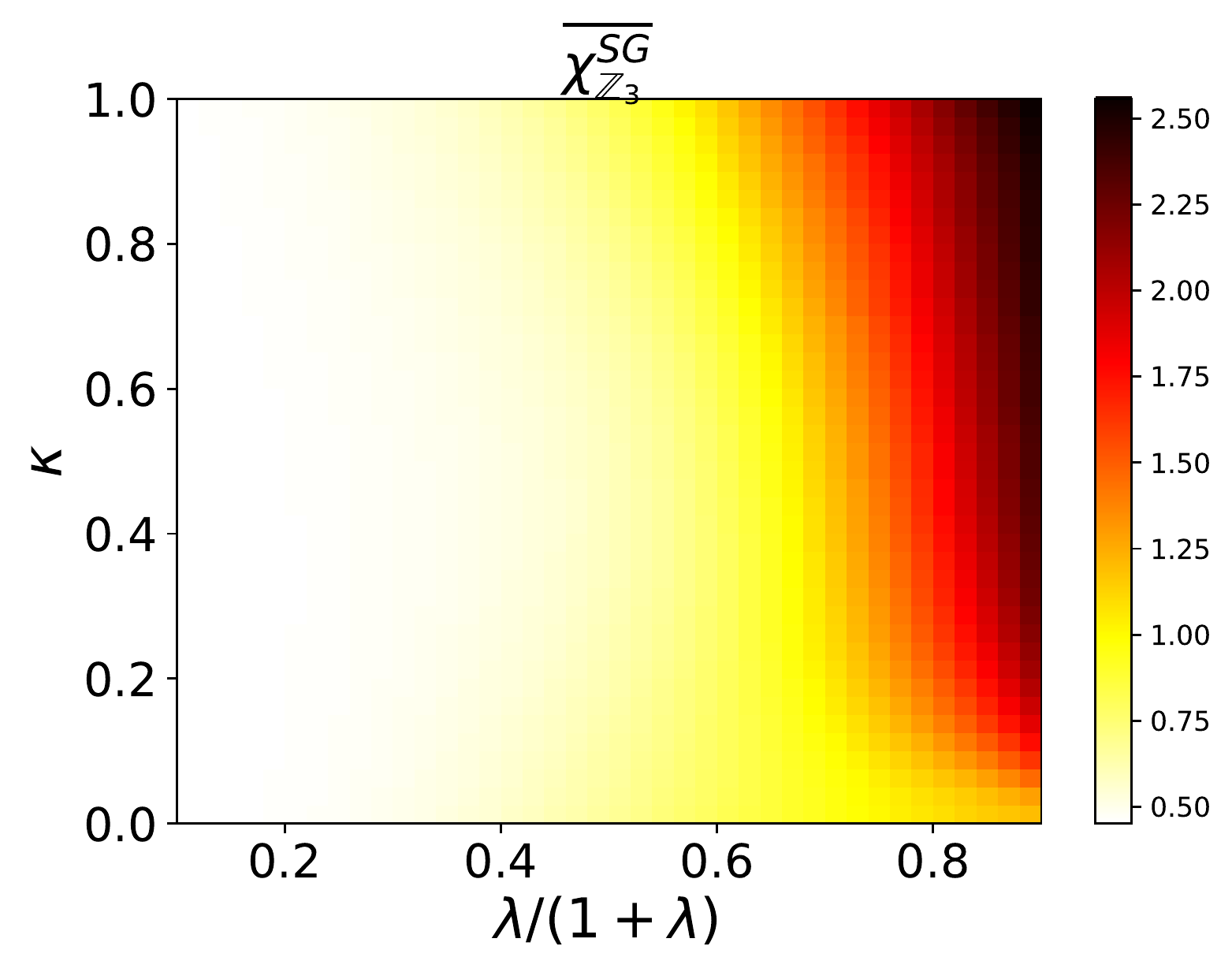}	&	 			
		\includegraphics[width=60mm]{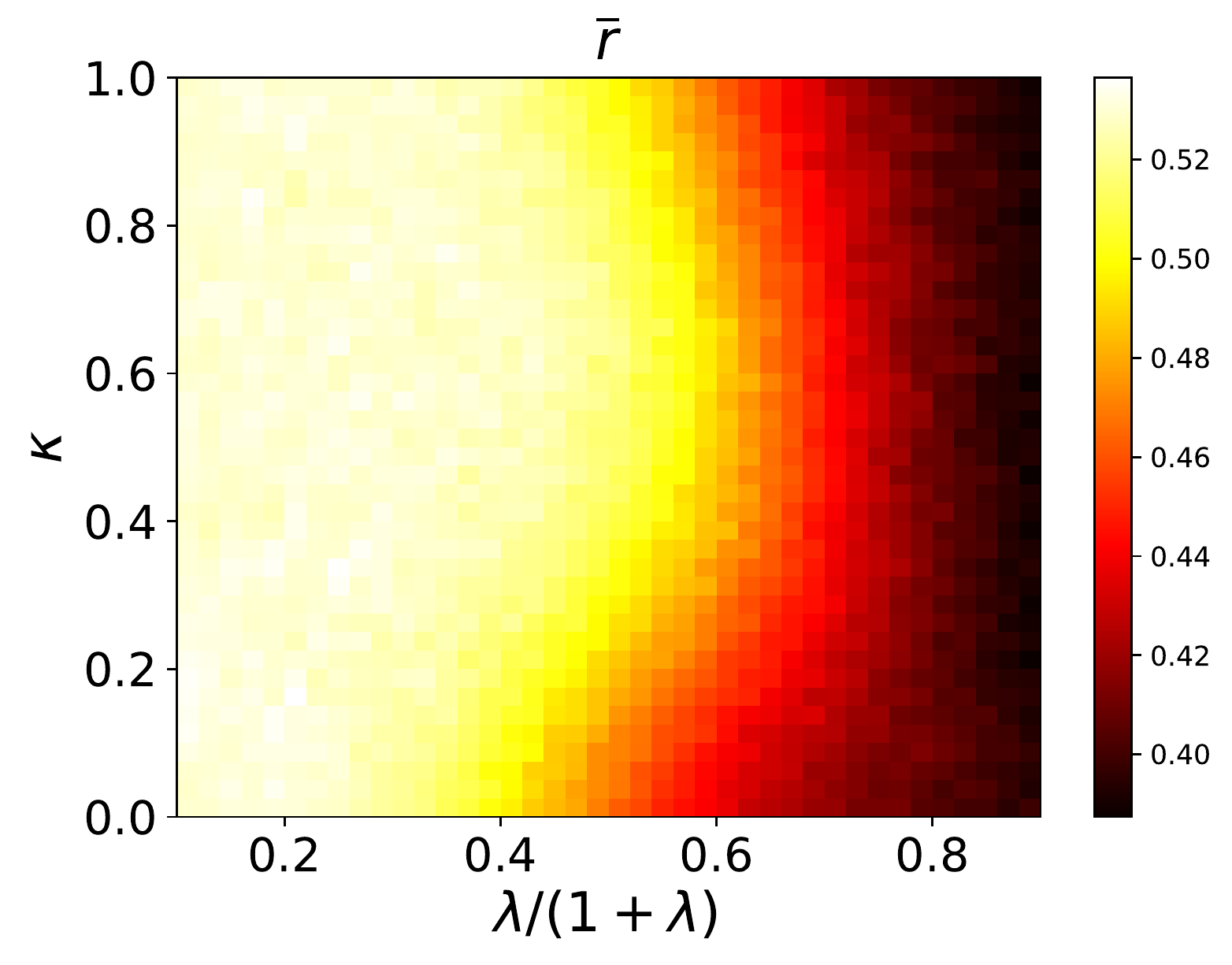}
	\end{tabular}		
	\caption{(Color online) Color map of $\overline{S'(L/4)}$, $\overline{\ea}$ and $\overline{r}$ for $\lambda/(1+\lambda) \in [0.1,0.9]$ and $\kappa \in [0,1]$.  200 eigenstates of the 7 site Hamiltonian~[\ref{eq:S3 Hamiltonian}] that transform as 1D irreps sampled randomly for 238 disorder realizations. \label{fig:phase digs}}
\end{figure*}
	\begin{figure}[!htb]
		\includegraphics[scale=0.35]{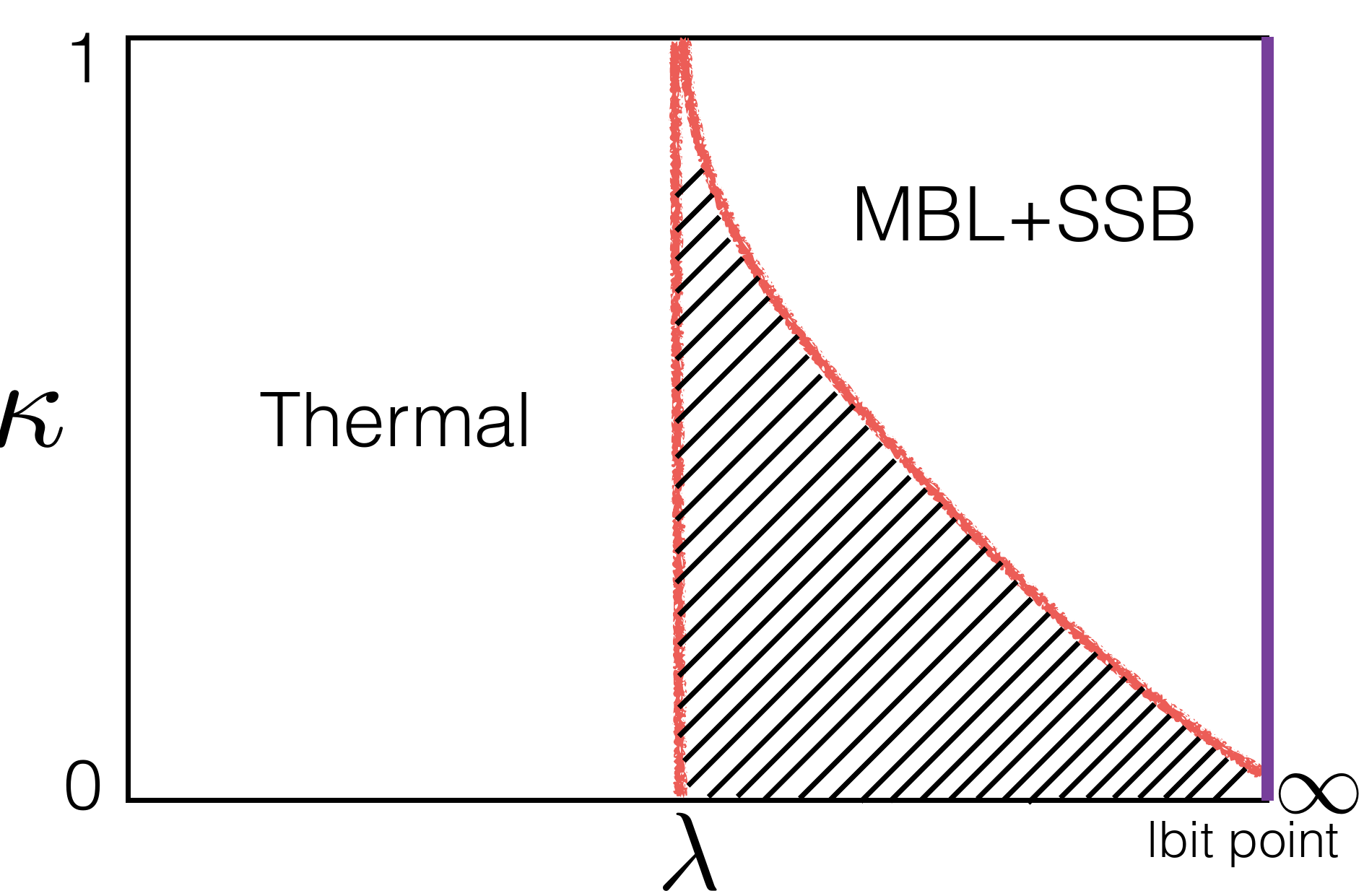}
		\caption{(Color online) The three regions labeled in a schematic plot. \label{fig:labeled phase diagram}}
	\end{figure}

{\it Thermal phase.} For the thermal phase, the distribution of slopes has a mean that saturates at the maximal entropy per unit site. This indicates a substantial presence of volume law scaling eigenstates. The distribution of $\ea$ has mean that does not increase with system size and implies absence of SSB. The level statistics clearly show Wigner-Dyson distribution highlighting the thermal nature of this regime.
	
{\it MBL+SSB phase.} For the MBL+ SSB phase, the distribution of slopes has mean close to 0. This indicates the presence of area-law scaling eigenstates. The distribution of $\ea$ parameter (designed to detect $S_3 \rightarrow \mathbb{Z}_3$ symmetry breaking) increases with system size. This indicates a substantial presence of eigenstates with SSB. The above results are consistent with a spin-glass phase with a residual Abelian symmetry group that supports a full MBL states with strong signatures of SSB at all eigenstates. The level statistics data is more noisy due to the level clustering as a result of degeneracies in this regime but hovers around the Poisson value.  
		
{\it Intermediate phase.} This case is shown as dashed region in the schematic Fig~\ref{fig:labeled phase diagram}.  The full $S_3$ symmetry is intact in this regime and is incompatible with the full MBL phase. At the l-bit point ($\lambda \rightarrow \infty$), the many-body spectrum has extensive number of states with extensive degeneracy. Any infinitesimal thermal perturbation is expected to split this extensive degeneracy down to the minimal values set by the size of the irreps of the group (2 and 1 in our case). This is expected to result in resonant delocalization of an extensive number of states. Thus the most plausible scenario is a fully thermal phase. Our numerical results, however, are not fully consistent with a thermal phase, but do not have the strong signatures of SG or MBL phases, either. A possibility for this regime is a marginal MBL phase that may be separating another MBL+SSB phase that is not explored within our parameter space. This can be uncovered introducing disorder to some of the relevant two-body thermal terms which are non-disordered in our current analysis. We leave this investigation for future work.}

{{Finite-size effects such as the critical-cone region~\cite{husekhemani2016critical} are important in this putative intermediate phase.  Indeed, such finite-size effects play an important role at large $\lambda$, and one has to be careful about the order of the thermodynamic and large $\lambda$ limits.  At finite-size $L$, large $\lambda$ corresponds to adding a small perturbation of $H_t$ on top of the lbit Hamiltonian, and for $1/\lambda$ much smaller than the lbit many-body level spacing $\sim 3^{-L}$ one can apply first-order perturbation theory.  One effect is to split the exponentially degenerate lbit states, resulting presumably in volume law eigenstates.  However, there are also states in the lbit spectrum which are not highly degenerate, corresponding to sectors where most of the sites sit in the one dimensional irrep, and only a few sit in the two dimensional irrep.  These are area law states, and could remain area law upon the addition of such a small perturbation.  Thus some vestiges of localization are expected to remain at large $\lambda$ for a finite-size $L$.  More numerical work at larger sizes will have to be done to distinguish these effects from a truly thermodynamic intermediate phase.}}

				 				 		
\section{Conclusion and outlook}
In this paper, we construct and study a spin-1 Hamiltonian invariant under an on-site $S_3$ symmetry using various disorder and symmetry diagnostics. We study the eigenstate phases that can arise via ergodicity and symmetry breaking. Within the accuracy of our numerical analysis, we can identify three regions in the two-parameter Hamiltonian space two of which are consistent with thermal, MBL- spin-glass phases and a third whose identity is not established with certainty in the present study. We state our observations about various characteristics of this region and speculate with regard to its identity. 

There are other interesting questions that are left for future study. In this work, we are limited to small system sizes by the tools of numerical analysis employed $i.e.$ exact diagonalization. There is much to be gained in designing and employing other numerical techniques to study larger system sizes. In this regard, it would be useful to explore the extension of tensor network techniques, which have been shown to be effective in the case of full-MBL, where eigenstates are expected to have area-law properties, to other settings, in particular to further study the disordered region with unbroken $S_3$ symmetry. Another extension of is to repeat our study in a Floquet setting where exotic possibilities have been conjectured in the presence of global symmetries~\cite{floquet_classification_shivaji,floquet_classification_vishwanath,floquet_classification_dominic,time_crystals_shivaji,time_crystals_dominic}. Here, the arguments for  instability of disordered Floquet systems with non-Abelian symmetries to thermalization would be stronger because of energy pumping. However, as we saw in the equilibrium case, such a system might have interesting features worth exploring. Furthermore, we could consider relaxing the setting of indefinitely stable phases, like MBL, to the so-called `pre-thermal' setting~\cite{prethermal_abanin2015asymptotic,prethermal_dominic} where the system is stable for times that are exponential in system size, and study the role of non-Abelian symmetries.

{\it Note added: } {As we were finishing this draft, we learned of a related paper: 
``Localization-protected order in spin chains with non-Abelian discrete symmetries'' (arXiv:1706.00022) by
Aaron J. Friedman, Romain Vasseur, Andrew C. Potter, S. A. Parameswaran}.

 \section*{Acknowledgements}
A.P would like to acknowledge the stimulating environment of the Boulder Summer School for Condensed Matter Physics-2016 which is supported by the NSF and the University of Colorado, where this project was initiated. We are grateful to Rahul Nandkishore, Bela Bauer, Sid Parameshwaran and Andrew Potter for helpful discussions. T.-C.W is supported by NSF Grants No. PHY 1314748 and No. 1620252. LF is supported by NSF DMR-1519579 and Sloan FG-2015-65244.

\bibliography{references}{}
\bibliographystyle{unsrt}

\appendix
\section{Incompatibility of non-Abelian symmetries with full MBL}~\label{sec:Potter review}
Here, we review the main hypothesis of Potter and Vasseur~\cite{potter2016symmetry}. The authors consider the case of a fully MBL system (as opposed to the case of a partially local, partially thermal system with mobility edges), and study the compatibility of MBL in the presence of various global symmetries. The working definition of an MBL system they consider is the existence of a complete set of quasi-local conserved quantities with associated quasi-local projectors in terms of which the Hamiltonian, $H$  can be defined as~\cite{SerbynAbanin2013LocalConservationLaws,huseNandkishoreOganesyan2014phenomenology,VoskAltman2013_MBLRGfixedpoint},
	\begin{multline}\label{eq:lbit hamiltonian}
		H = \sum_{i=1}^L \sum_{\alpha = 1}^D E[i]_{\alpha} \hat{P}[i]_{\alpha} \\	+ \sum_{i \neq j =1}^L \sum_{\alpha,\beta = 1}^D E[i,j]_{\alpha,\beta} \hat{P}[i]_\alpha \hat{P}[j]_\beta \\ + \sum_{i \neq j \neq k=1}^L \sum_{\alpha,\beta, \gamma = 1}^D E[i,j,k]_{\alpha,\beta,\gamma} \hat{P}[i]_\alpha \hat{P}[j]_\beta \hat{P}[k]_\gamma + \ldots
	\end{multline}
Here, $L$ is the number of spins, $\hat{P}[i]_\alpha$ is the projector onto the $\alpha^{th}$ quasi-local conserved quantity at the $i^{th}$ location, $D$ is the number of conserved quantities and the $E$'s are constants that fix the energy eigenvalues. Furthermore, we can apply a finite depth quantum unitary circuit,  $W$, that re-expresses the conserved quantities as local degrees of freedom. These local objects are called l-bits in terms of which the original spins (p-bits~\footnote{Note that we do not restrict ourselves to a local two-dimensional Hilbert space. This means, we should be talking about p-dits and l-dits instead of p-bits and l-bits. However, to keep with standard terminology, we will use the latter names.}) and Hamiltonian can be defined. This is typically called the l-bit Hamiltonian, $H_{lbit}$. 

Let us now consider the case when the Hamiltonian $H$ has a global on-site symmetry $i.e$  $H$ commutes with the unitary representation of some group, $G$ of the form $U(g) = \bigotimes_{i=1}^L V_i(g)$, where, $V_i(g)$ acts on each physical spin.
\begin{equation}
g:H \rightarrow U(g) H U^\dagger(g) = H
\end{equation}
In this case, representation theory of the group $G$ plays a role in constraining the allowed form of the l-bit Hamiltonian~[\ref{eq:lbit hamiltonian}]. For spins (p-bits) to allow a well defined on-site group action, the local Hilbert space of each spin must correspond to some faithful representation of that group. This must also be true for a tensor product of the p-bits that constitute an l-bit. In other words, we can write the l-bit- projectors $\hat{P}[i]_\alpha$ using a fully reduced basis that can be labeled as $\ket{\Gamma, m_\Gamma; d_\Gamma}$
\begin{equation}
\hat{P}_{\Gamma, m_\Gamma; d_\Gamma} = \ket{\Gamma, m_\Gamma; d_\Gamma} \bra{\Gamma, m_\Gamma; d_\Gamma}
\end{equation}
We define the different labels below and compare them with the well known case of the representations of the rotation group $SO(3)$: 
\begin{itemize}
	\item $\Gamma = 1\ldots N_R$ labels the irreducible representation (irrep) of the group and is equivalent to the total angular momentum quantum number, $j$. The number of values it can take is equal to the number of irreps of $G$, $N_R$.
	\item $m_\Gamma = 1\ldots |\Gamma|$  is equivalent to the azimuthal quantum number $m_j$. The number of values it can take is equal to the dimension of the irrep, $|\Gamma|$. 
	\item $d_\Gamma = 1 \ldots D_\Gamma$ labels which of the $D_\Gamma$ copies of the $\Gamma$ irrep is being considered.
\end{itemize}    

 In this basis, the action of the group is
 \begin{equation}
g:\ket{\Gamma, m_\Gamma; d_\Gamma} \rightarrow \sum_{n_\Gamma = 1}^{|\Gamma|}\Gamma(g)_{n_\Gamma, m_\Gamma} \ket{\Gamma, n_\Gamma; d_\Gamma}
 \end{equation}
Demanding the invariance of the Hamiltonian~[\ref{eq:lbit hamiltonian}] under group action and invoking Schur's lemma~\cite{ramond_group} irrep-wise, we get the constrained form of the Hamiltonian compatible with on-site symmetry as
	\begin{multline}\label{eq:symmetric_lbit hamiltonian}
	H = \sum_{i=1}^L \sum_{\Gamma = 1}^{N_R} \sum_{d^i_\Gamma =1}^{D^i_\Gamma} E[i]_{\Gamma,d^i_\Gamma} \hat{P}[i]_{\Gamma,d^i_\Gamma} \\+ \sum_{i \neq j =1}^L \sum_{\Gamma, \Gamma' = 1}^{N_R} \sum_{d^i_\Gamma =1}^{D^i_\Gamma} \sum_{d^j_{\Gamma'} =1}^{D^j_{\Gamma'}} E[i,j]_{\Gamma,d^i_{\Gamma},\Gamma',d^j_{\Gamma'}} \hat{P}[i]_{\Gamma,d^i_{\Gamma}} \hat{P}[j]_{\Gamma',d^j_{\Gamma'}} \\ + \ldots
	\end{multline}
	where, 
	\begin{equation}
	\hat{P}_{\Gamma;d_\Gamma} \equiv \sum_{m_\Gamma = 1}^{|\Gamma|} \outerproduct{\Gamma,m_\Gamma;d_\Gamma}{\Gamma,m_\Gamma;d_\Gamma}.
	\end{equation}
	We now consider the cases of Abelian and non-Abelian symmetries separately. If $G$ is an Abelian group, all irreps are one dimensional ($|\Gamma|=1$). This means that all projectors $P_{\Gamma;d_\Gamma}$ are rank-1 which preserves the form~[\ref{eq:lbit hamiltonian}]. With sufficient disorder, resulting in sufficiently random $E's$, we can imagine that all degeneracies are lifted and we can obtain an MBL phase stable to perturbations. However, if $G$ is non-Abelian, not all irreps are one dimensional. This means that we invariably have higher-rank local projectors giving us only a partial set of conserved quantities rather than complete which leads to degeneracies that are extensive in system size. This is clearly seen by examining the eigenvalues and eigenvectors of Eq~[\ref{eq:symmetric_lbit hamiltonian}].
\begin{eqnarray}
&&H \ket{\{\Gamma,m_\Gamma;d_\Gamma\}} = \mathcal{E}(\{\Gamma;d_\Gamma\}) \ket{\{\Gamma;d_\Gamma\}}, \nonumber \\
&&\mathcal{E}(\{\Gamma;d_\Gamma\}) = \sum_{i=1}^L E[i]_{\Gamma_i,d^i_{\Gamma_i}} + \sum_{i \neq j =1}^L E[i,j]_{\Gamma_i,d^i_{\Gamma_i},\Gamma'_j,d^j_{\Gamma'_j}}  + \ldots \nonumber\\
&&\ket{\{\Gamma,m_\Gamma;d_\Gamma\}} = \bigotimes_{i=1}^L \ket{\Gamma_i,m_{\Gamma_i};d^i_{\Gamma_i}} .
\end{eqnarray}
Note that none of the eigenvalues, $\mathcal{E}$ have any labels corresponding to the inner multiplicity of the irreps $m_\Gamma$. This means that the eigenstate $\ket{\{\Gamma,m_\Gamma;d_\Gamma\}}$ has a degeneracy of $|\Gamma_1| \times |\Gamma_2| \times \ldots |\Gamma_L|$ which is clearly extensive in system size when $G$ is non-Abelian. The authors of Ref~\onlinecite{potter2016symmetry} state that under the influence of perturbations, such a degeneracy is susceptible to long range resonances which destabilizes MBL. Furthermore, they suggest a possible set of phases for the system to be in depending on the nature of the global symmetry group $G$. Here we list the possibilities for the case of finite groups such as  $S_{n \ge 3}$ and $D_{n \ge 3}$, (i) Ergodic/ thermal phase, (ii) The so-called MBL spin-glass (SG)~\cite{khemaniSondhi2016floquet} phase characterized by localization with symmetry spontaneously broken (SSB) to an Abelian subgroup, and (iii) The so-called quantum critical glass phase (QCG)~\cite{glass_sidPotterPRL2015,glass_pekkerRefael_PRX2014,glass_kang2016} characterized by critical scaling of entanglement entropy. Within the accuracy of our numerical analysis, our findings are consistent with the conjecture of Potter and Vasseur.
\section{Constructing symmetric Hamiltonians}\label{app:building symmetric ham}
In this appendix, we give details of how the Hamiltonian used in the main text, Eq~[\ref{eq:S3 Hamiltonian}] was constructed. We also detail a general technique to construct local symmetric operators with which we can build spin Hamiltonians invariant under any on-site symmetry in any dimension. The construction of the 1D $S_3$ invariant Hamiltonian of the main text is a specific application of this general technique.

\subsection{Building the $S_3$ invariant Hamiltonian}
\subsubsection{Basis properties of $S_3$ and its representation used.}
We first review some basic properties of the group $S_3$ and its representation used in this paper. $S_3$, the symmetry group of three objects is the smallest non-Abelian group. It is of order 6 and can be generated using two elements and the following presentation
\begin{equation}
\innerproduct{a,x}{a^3=x^2=1, xax = a^{-1}}.
\end{equation}
It has three irreducible representations, $\mathbf{1}, \mathbf{1}', \mathbf{2}$ which can be written as 
\begin{enumerate}
	\item  $\chi^{\mathbf{1}}(a) = 1$, $\chi^{\mathbf{1}}(x) = 1$,
	\item  $\chi^{\mathbf{1'}}(a) = 1$, $\chi^{\mathbf{1'}}(x) = -1$,
	\item  $\Gamma^\mathbf{2}(a)= \begin{pmatrix}
	\omega & 0 \\
	0 & \omega^*
	\end{pmatrix}$, $\Gamma^\mathbf{2}(x)= \begin{pmatrix}
	0 & 1 \\
	1 & 0
	\end{pmatrix}$ ,
\end{enumerate}
where $\omega = e^{2 \pi i/ 3}$. The local Hilbert space we have chosen for each spin lives is the three-dimensional reducible representation $\mathbf{2} \oplus \mathbf{1}'$. We use the eigenspace of the spin-1 angular momentum operator $S^z$ to label the irreps. The $\mathbf{2}$ irrep is encoded in the two-dimensional $\pm1$ eigenspace of $S^z$ (which we will call $\ket{\pm}$) and the $\mathbf{1}'$ is encoded in the one-dimensional $0$ eigenspace of $S^z$ (which we will call $\ket{0}$). The matrix representation of a general group element in this basis looks like
\begin{equation}
V(g) = \begin{pmatrix}
\Gamma^\mathbf{2}(g)_{11}& 0 & \Gamma^\mathbf{2}(g)_{12}\\
0 & \chi^{\mathbf{1'}}(g) & 0\\
\Gamma^\mathbf{2}(g)_{21} & 0 & \Gamma^\mathbf{2}(g)_{22}.
\end{pmatrix}
\end{equation}
In particular, the generators have the following representation
\begin{eqnarray}
V(a) = \begin{pmatrix}
\omega & 0 & 0\\
0 & 1 & 0\\
0 & 0 & \omega^*
\end{pmatrix},~
V(x) = \begin{pmatrix}
0 & 0 & 1\\
0 & -1 & 0\\
1 & 0 & 0
\end{pmatrix}.
\end{eqnarray}
\subsubsection{1-spin operator}
We first construct a 1-spin invariant operator. We start with a general operator that acts on the space of a single spin 
\begin{equation}
\hat{\lambda} = 
\begin{pmatrix}
\lambda_{11} & \lambda_{12} & \lambda_{13} \\
\lambda_{21} & \lambda_{22} & \lambda_{23} \\
\lambda_{31} & \lambda_{32} & \lambda_{33} \\
\end{pmatrix}
\end{equation}
We now demand that $\hat{\lambda}$ is invariant under conjugation by $V(g)$, $i.e.$, the representation of symmetry on a single spin. 
\begin{equation}
V(g) \hat{\lambda} V^\dagger(g) = \hat{\lambda}
\end{equation}
Schur's lemma~\cite{ramond_group} constrains the matrix elements of $\hat{\lambda}$ in the following way:
\begin{enumerate}
	\item $\hat{\lambda}$ cannot mix basis states corresponding to \emph{different} irreps.  
	\item $\hat{\lambda}$ must be proportional to the identity operator when acting on basis states corresponding to the \emph{internal} states of the \emph{same} irrep. 
	\item If there are \emph{multiple copies} of the same irrep, $\hat{\lambda}$ can mix the basis states corresponding the \emph{same internal state}  of different copies but should still be proportional to the identity operator as an action on the internal states. 
\end{enumerate}	
The meaning of these constraints should become clearer with the applications that will follow. For a single spin operator, since we have only one copy of each irrep, constraint 3 does not apply. Applying constraints 1 and 2, we get the form of $\hat{\lambda}$:
\begin{multline}
\hat{\lambda} = \begin{pmatrix}
\lambda_{\mathbf{2}} & 0 & 0 \\
0 & \lambda_{\mathbf{1'}} & 0 \\
0 & 0 & \lambda_{\mathbf{2}}
\end{pmatrix}  = \lambda_{\mathbf{1'}} \outerproduct{0}{0} + \lambda_{\mathbf{2}} \left(\outerproduct{+}{+}+\outerproduct{-}{-} \right)  \\= (\lambda_{\mathbf{2}}-\lambda_{\mathbf{1'}}) (S^z)^2 + \lambda_{\mathbf{1'}} \mathbb{1}_3
\end{multline}
From this, we can read off the only non-trivial 1-spin symmetric operator, $(S^z)^2$ which is also Hermitian.

\subsubsection{2-spin operator}
In order to find a symmetric 2-spin operator, we follow the same logic as that of a 1-spin operator. First, we start with a general operator that acts on the 9 dimensional vector space of 2 spins,
\begin{equation}
\hat{J} = \begin{pmatrix}
J_{11} & J_{12} & \ldots & J_{19} \\
J_{21} & J_{22} & \ldots & J_{29} \\
\vdots &  & \ddots & \vdots \\
J_{91} & J_{92} & \ldots & J_{99}.
\end{pmatrix}
\end{equation}
Next, we demand invariance under conjugation by $V(g) \otimes V(g)$, $i.e.$, the representation of symmetry on two spins. 
\begin{equation}
V(g) \otimes V(g) \hat{J} V^\dagger(g) \otimes V^\dagger(g) = \hat{J}
\end{equation}
We now have to impose the constraints coming from Schur's lemma. However, in order to do that, we need to find out the irrep content of $V(g) \otimes V(g)$. For this, we first list the Clebsch Gordan (CG) decomposition that gives us the outcomes of fusing different $S_3$ irreps. This is the generalization of angular momentum addition of $SU(2)$ irreps. Note that we exclude the trivial case of fusion with the trivial irrep $\mathbf{1}$,
\begin{eqnarray}
\mathbf{1'} \otimes \mathbf{1'} &\cong& \mathbf{1} \nonumber \\
\mathbf{2}~ \otimes \mathbf{1'} &\cong& \mathbf{2} \nonumber \\
\mathbf{2}~ \otimes \mathbf{2}~ &\cong& \mathbf{2} \oplus \mathbf{1'} \oplus \mathbf{1} \nonumber
\end{eqnarray}
The irrep content of $V(g) \otimes V(g)$ is obtained from the CG decomposition,
\begin{equation}
(\mathbf{2} \oplus \mathbf{1'}) \otimes (\mathbf{2} \oplus \mathbf{1'}) \cong \mathbf{1} \oplus \mathbf{1} \oplus \mathbf{1'} \oplus \mathbf{2} \oplus \mathbf{2} \oplus \mathbf{2}.
\end{equation}
 It is clear that we are bound to have multiple copies of the same irrep in this decomposition for which, using Constraint 3 imposed by Schur's lemma, unlike the single spin case, we can get off diagonal operators. Let us list and label the different instances of each irrep appearing in the decomposition for convenience.
\begin{eqnarray}
\mathbf{1'} &\otimes& \mathbf{1'} \rightarrow \mathbf{1}_A \nonumber \\
\mathbf{2} &\otimes& \mathbf{2} \rightarrow \mathbf{1}_B \nonumber \\
\mathbf{2} &\otimes& \mathbf{2} \rightarrow \mathbf{1'} \nonumber \\
\mathbf{1'} &\otimes& \mathbf{2} \rightarrow \mathbf{2}_A\nonumber \\
\mathbf{2} &\otimes& \mathbf{1'} \rightarrow \mathbf{2}_B\nonumber \\
\mathbf{2} &\otimes& \mathbf{2} \rightarrow \mathbf{2}_C \nonumber 
\end{eqnarray}
The subscripts label the copy of the irrep. We next need the basis of the representation of each irrep in $V(g) \otimes V(g)$. These can be written in terms of the original basis states (labeled by $S^z$ eigenvalues) using CG coefficients which we calculate using the technique by Sakata~\cite{sakata} (also see \cite{abhishodh_pra}). 
\begin{eqnarray}
\ket{\mathbf{1}_A} &=& \ket{0} \ket{0} \nonumber \\
\ket{\mathbf{1}_B} &=& \frac{\ket{+}\ket{-}+\ket{-}\ket{+}}{\sqrt{2}} \nonumber\\
\ket{\mathbf{1'}}&=& \frac{\ket{+}\ket{-}-\ket{-}\ket{+}}{\sqrt{2}} \nonumber\\
\ket{\mathbf{2}_A,\pm} &=& \pm \ket{0} \ket{\pm} \nonumber \\
\ket{\mathbf{2}_B,\pm} &=& \pm \ket{\pm} \ket{0} \nonumber \\
\ket{\mathbf{2}_C,\pm} &=& \ket{\mp}  \ket{\mp} \label{eq:S3 Cg}
\end{eqnarray}
Using this, we have the 2-spin  $S_3$ symmetric operator constrained by Schur's lemma
\begin{multline} \label{eq:S3 2body symm ketform}
\hat{J} = J^\mathbf{1'} \outerproduct{\mathbf{1'}}{\mathbf{1'}}  + \sum_{\mu, \nu = A,B} J^\mathbf{1}_{\mu \nu} \outerproduct{\mathbf{1}_\mu}{\mathbf{1}_\nu}  \\
+\sum_{\mu, \nu = A,B,C} J^\mathbf{2}_{\mu \nu} (\outerproduct{\mathbf{2}_\mu, +}{\mathbf{2}_\nu, +}  + \outerproduct{\mathbf{2}_\mu, -}{\mathbf{2}_\nu, -} )
\end{multline}
As in the case of 1-spin invariant operator, we can again read off the several independent symmetric 2-spin operators by simplifying  Eq~.\ref{eq:S3 2body symm ketform}. However, since we need the operators to be Hermitian, we take Hermitian combinations of these operators. We finally list the non-trivial independent Hermitian operators expressed in terms of spin-1 operators.
\begin{eqnarray}
\hat{J}_1 &=& S^z \otimes S^z,\nonumber \\
\hat{J}_2 &=& (S^z)^2 \otimes (S^z)^2 \nonumber \\
\hat{J}_3 &=& (S^+)^2 \otimes (S^-)^2 + (S^-)^2 \otimes (S^+)^2 \nonumber \\
\hat{J}_4 &=& (S^+ S^z) \otimes (S^- S^z) + (S^- S^z) \otimes (S^+ S^z)  + h.c \nonumber \\
\hat{J}_5 &=&(S^- S^z) \otimes (S^z S^+) + (S^+ S^z) \otimes (S^z S^-) + h.c \nonumber \\
\hat{J}_6 &=& (S^+ S^z) \otimes (S^+)^2 + (S^- S^z) \otimes (S^-)^2 + h.c \nonumber \\	 
\hat{J}_7 &=&(S^+)^2 \otimes (S^+S^z) + (S^-)^2 \otimes (S^-S^z) + h.c \nonumber 
\end{eqnarray}
To construct the Hamiltonian~.\ref{eq:S3 Hamiltonian}, we have used $(S^z)^2$ and $\hat{J}_1$ to build the disordered part $H_d$ and $\hat{J}_3$, $\hat{J}_4$, $\hat{J}_7$ to build the thermal part of the Hamiltonian, $H_t$. Since $\hat{J}_5$ and $\hat{J}_6$ are mirrored versions of $\hat{J}_4$ and $\hat{J}_7$ respectively, we can keep our thermal term sufficiently generic even if leave them out. Note that while we have constructed a Hamiltonian for a 1D spin chain, the symmetric operators constructed using this formalism can be used to build Hamiltonians for any spatial dimensions.

\subsection{General technique}
We now give details of a general procedure that can be applied to obtain n-spin symmetric operators invariant under a representation of any on-site symmetry group. The schematic procedure is as follows
\begin{enumerate}
	\item To construct symmetric n-spin operators, we first write down general operators that act on the Hilbert space of n spins. 
	\item We then demand the invariance of this operator under symmetry action.
	\item Using Schur's lemma, we constrain the matrix elements of the n-spin operator and read off the independent operators.
	\item If required, we finally take the hermitian combinations of the independent operators. 
\end{enumerate}
If $G$ is the group we are considering, the most general local Hilbert space for the spin compatible with G-action can be an arbitrary number of copies of each irrep of $G$. Like mentioned in Sec~\ref{sec:Potter review}, we can choose the basis as	$\ket{\Gamma, m_\Gamma; d_\Gamma}$ where the symbols have the same meaning as before.  The matrix representation of the group operators in this basis is block diagonal.
\begin{eqnarray}
	U(g) &=& \bigotimes_{i=1}^L V_i(g) \\
	V_i(g) &=& \bigoplus_{\Gamma=1}^{N_R} \mathbb{1}_{D_\Gamma^i} \otimes \Gamma(g)
\end{eqnarray}
The ``passive" group action on the basis is
\begin{equation}
	g:\ket{\Gamma, m_\Gamma; d_\Gamma} \mapsto \sum_{n_\Gamma = 1}^{|\Gamma|}\Gamma(g)_{n_\Gamma, m_\Gamma} \ket{\Gamma, n_\Gamma; d_\Gamma}
\end{equation}

\subsubsection{1-spin symmetric operator}
Let us now start with symmetric 1-spin operators. The most general 1-spin operator we can write down is 
\begin{equation}\label{eq:general 1 spin operator}
	\hat{\lambda} =  \lambda^{\Gamma, m_\Gamma; d_\Gamma}_{\Gamma', m_{\Gamma'}; d_{\Gamma'}} \outerproduct{\Gamma, m_\Gamma; d_\Gamma}{\Gamma', m_{\Gamma'}; d_{\Gamma'}}
\end{equation}
Note that for notational convenience, here and henceforth, we assume summation over repeated indices. Demanding invariance under conjugation with $V(g)$, we have 
\begin{eqnarray} \label{eq:1 body invariance}
	V(g) \lambda V(g)^\dagger &=&  \lambda \\
\implies	  \Gamma(g)_{m_\Gamma,n_\Gamma} (\Gamma'(g)_{m_{\Gamma'},n_{\Gamma'}})^* \lambda^{\Gamma, n_\Gamma; d_\Gamma}_{\Gamma', n_{\Gamma'}; d_{\Gamma'}} &=& \lambda^{\Gamma, m_\Gamma; d_\Gamma}_{\Gamma', m_{\Gamma'}; d_{\Gamma'}} \nonumber.
\end{eqnarray}
In matrix form, the condition on $\lambda$ becomes
\begin{equation}
	[\Gamma(g)] [\lambda_{\Gamma';d_{\Gamma'}}^{\Gamma;d_{\Gamma}}] = 	 [\lambda_{\Gamma';d_{\Gamma'}}^{\Gamma;d_{\Gamma}}] [\Gamma'(g)].
\end{equation}
This means that $[\lambda_{\Gamma';d_{\Gamma'}}^{\Gamma;d_{\Gamma}}]$ is an intertwiner between the irreps $\Gamma$ and $\Gamma'$. Such a matrix is constrained by Schur's lemma:
\begin{eqnarray}
	[\lambda_{\Gamma';d_{\Gamma'}}^{\Gamma;d_{\Gamma}}] = 0 &\text{      if      }& \Gamma \neq \Gamma' \nonumber \\
	\propto \mathbb{1}_\Gamma &\text{     if     }& \Gamma = \Gamma' \nonumber \\
	\implies \lambda^{\Gamma, m_\Gamma; d_\Gamma}_{\Gamma', m'_{\Gamma'}; d'_{\Gamma'}} = 0 &\text{      if      }& \Gamma \neq \Gamma' \nonumber \\
	\propto \delta_{m_\Gamma,m'_{\Gamma'}}  &\text{     if     }& \Gamma = \Gamma' 
\end{eqnarray}

Using this in Eq~\ref{eq:general 1 spin operator}, we get the form of a symmetric 1-spin operator,
\begin{equation} \label{eq:1 body symmetric}
\hat{\lambda} =  \lambda^{\Gamma}_{d_\Gamma,d'_\Gamma} \outerproduct{\Gamma,m_\Gamma;d_\Gamma}{\Gamma,m_\Gamma;d'_\Gamma}.
\end{equation}
Here, the non-zero matrix elements $\lambda^{\Gamma}_{d_\Gamma,d'_\Gamma}$ are free parameters that act on the degenerate subspace associated with the outer multiplicity of each irrep.
\subsubsection{2-spin symmetric operator}
We now consider 2-spin symmetric operators. The most general operator that acts on the 2-spin Hilbert space is 
\begin{multline} \label{eq:2 body general}
\hat{J} = J^{({\Delta},m_{\Delta};d_{\Delta}),({\Lambda},m_{\Lambda};d_{\Lambda})}_{({\Delta'},m_{\Delta'};d_{\Delta'}),({\Lambda'},m_{\Lambda'};d_{\Lambda'})} \outerproduct{{\Delta},m_{\Delta};d_{\Delta}}{{\Delta'},m_{\Delta'};d_{\Delta'}} \\ \outerproduct{{\Lambda},m_{\Lambda};d_{\Lambda}}{{\Lambda'},m_{\Lambda'};d_{\Lambda'} }
\end{multline}
 For this operator to be symmetric, it needs to be invariant under conjugation by $V_{1}(g) \otimes V_{2}(g)$,
\begin{eqnarray} \label{eq:2body_invariance}
\left[V_{1}(g) \otimes V_{2}(g)\right] J \left[V_{1}(g) \otimes V_{2}(g)\right]^\dagger = J.
\end{eqnarray}  
To use the techniques like we did for the 1-spin operator in the previous subsection, we need to first block-diagonalize $V_{1}(g) \otimes V_{2}(g)$ using a suitable basis change and then use Schur's lemma. This redefinition of the basis states can be done using Clebsch-Gordan (CG) coefficients. Recall that the irrep content of a direct product of two irreps is schematically given by the CG series:
\begin{equation} \label{eq:CG series}
\Delta \otimes \Lambda = \bigoplus_\Gamma N_{\Delta \Lambda}^\Gamma~ \Gamma
\end{equation}
$N_{\Delta \Lambda}^\Gamma$ denotes the number of copies of the $\Gamma$ irrep that exists in the fusion outcome of $\Delta \otimes \Lambda$. At the level of representations, Eq~.\ref{eq:CG series} tells us that there exists a change of basis by a unitary matrix $C_{\Delta \Lambda}$ that fully reduces the direct product of the irreps $\Delta \otimes \Lambda$:
\begin{eqnarray}
\Delta(g) \otimes \Lambda(g) \cong \bigoplus_{\Gamma} \mathbb{1}_{N_{\Delta \Lambda}^\Gamma} \otimes \Gamma(g),\\
	C_{\Delta \Lambda} \left(\Delta(g) \otimes \Lambda(g)\right) C_{\Delta \Lambda}^\dagger = \bigoplus_\Gamma \mathbb{1}_{N_{\Delta \Lambda}^\Gamma} \otimes \Gamma(g),\\
	C^{\Gamma;\alpha_\Gamma}_{\Delta \Lambda} \left(\Delta(g) \otimes \Lambda(g)\right) C^{\Gamma;\alpha_\Gamma \dagger}_{\Delta \Lambda} = \Gamma(g).
	\end{eqnarray}
$C^{\Gamma;\alpha_\Gamma}_{\Delta \Lambda}$ are isometries whose matrix elements, $[C^{\Gamma;\alpha_\Gamma}_{\Delta \Lambda}]_{m_\Delta,m_\Lambda}^{m_\Gamma}$ are the CG coefficients that project $\Delta \times \Lambda$ onto the $\alpha_\Gamma^{th}$ copy of the irrep $\Gamma$, where, $\alpha_\Gamma = 1 \ldots N_{\Delta \Lambda}^\Gamma$. It is useful to look at the `passive' action of the CG coefficients on the basis kets to remind us of the actual ``change-of-basis" action,
\begin{equation}\label{eq:cg basis action}
[C^{*\Gamma;\alpha_\Gamma}_{\Delta \Lambda}]_{m_\Delta,m_\Lambda}^{m_\Gamma} \ket{\Delta,m_\Delta} \ket{\Lambda,m_\Lambda} = \ket{\Gamma,m_\Gamma,\alpha_\Gamma}
\end{equation}
 Now, consider the following equivalences
\begin{multline} \label{eq:reduce 2 body symmetry}
V_1(g) \otimes V_2(g) = \left[\bigoplus_\Delta \mathbb{1}_{D^1_\Delta} \otimes \Delta(g) \right] \otimes \left[ \bigoplus_\Lambda \mathbb{1}_{D^2_\Lambda} \otimes \Lambda(g) \right] \\ 
\cong \bigoplus_{(\Delta,\Lambda)} \left[\mathbb{1}_{D^1_\Delta \times D^2_\Lambda}\right] \otimes \left[\Delta(g) \otimes \Lambda(g) \right]\\
\cong \bigoplus_{(\Delta,\Lambda)} \left[\mathbb{1}_{D^1_\Delta \times D^2_\Lambda}\right] \otimes \left[\bigoplus_{\Gamma} \mathbb{1}_{N^\Gamma_{\Delta \Lambda}} \otimes \Gamma(g) \right]\\
\cong \bigoplus_{\Gamma} \mathbb{1}_{\mathcal{D}^{1,2}_{\Gamma}} \otimes \Gamma(g).
\end{multline}
Here,
\begin{equation}
\mathcal{D}^{1,2}_\Gamma = \sum\limits_{(\Delta,\Lambda)|\Gamma \in \Delta \otimes \Lambda} D^1_\Delta~ D^2_\Lambda ~N^\Gamma_{\Delta \Lambda}
\end{equation}
  is the number of `fusion channels' of the kind $\Delta \otimes \Lambda \rightarrow \Gamma$ available to produce the irrep $\Gamma$ using the irreps in the 2-spin Hilbert space. In short, Eq~[\ref{eq:reduce 2 body symmetry}] tells us that there exists some unitary matrix $W$ which block diagonalizes $V_1(g) \otimes V_2(g)$
\begin{eqnarray}
	W \left[V_1(g) \otimes V_2(g)\right] W^\dagger = \tilde{V}(g)\\
	\tilde{V}(g) =  \bigoplus_{\Gamma \in \Delta \otimes \Lambda} \mathbb{1}_{\mathcal{D}^{1,2}_{\Gamma}} \otimes \Gamma(g).
\end{eqnarray}
When viewed as a matrix, $W$ contains operations to both reorder indices appropriately as well as use the CG coefficients to reduce the direct product of irreps form $V_1$ and $V_2$ block-by-block. If we operate $W$ on both sides of the Eq~[\ref{eq:2body_invariance}] and call $WJW^\dagger = K$, we get
\begin{equation} \label{eq:K symmetry}
\tilde{V}(g) K \tilde{V}(g)^\dagger = K
\end{equation}
The matrix elements of $K$ can be written in terms of those of $J$ and CG coefficients.
\begin{multline}
K^{\Gamma, m_\Gamma; c_\Gamma}_{\Gamma', m_{\Gamma'}; c_{\Gamma'}} = \left[C^{\Gamma;\alpha_\Gamma}_{\Delta \Lambda}\right]^{m_\Gamma}_{m_\Delta,m_\Lambda} \left[C^{*\Gamma';\alpha_\Gamma'}_{\Delta' \Lambda'}\right]^{m_{\Gamma'}}_{m_{\Delta'},m_{\Lambda'}} \\ J^{({\Delta},m_{\Delta};d_{\Delta}),({\Lambda},m_{\Lambda};d_{\Lambda})}_{({\Delta'},m_{\Delta'};d_{\Delta'}),({\Lambda'},m_{\Lambda'};d_{\Lambda'})} 
\end{multline}
Eq~[\ref{eq:K symmetry}] is of the same form as Eq~[\ref{eq:1 body invariance}] and we can use Schur's lemma again to constrain the elements of $K$, 
\begin{eqnarray}
K^{\Gamma, m_\Gamma; c_\Gamma}_{\Gamma', m_{\Gamma'}; c_{\Gamma'}} = 0 &\text{      if      }& \Gamma \neq \Gamma' \nonumber\\
\propto \delta_{m_\Gamma,m_{\Gamma'}}  &\text{     if     }& \Gamma = \Gamma'.
\end{eqnarray}
Note that $c_\Gamma$ is a collective index of compatible $\left(d_\Delta,d_\Lambda,\alpha_\Gamma\right)$ that runs over the $\mathcal{D}_\Gamma$ different fusion channels mentioned above. Also, note that we use the convention $\left[C^{\Gamma;\alpha_\Gamma}_{\Delta \Lambda}\right]^{m_\Gamma}_{m_\Delta,m_\Lambda}=0$ if $\Gamma \notin \Delta \otimes \Lambda$.  Finally, we can undo the transformation of $W$ and get the elements of $\hat{J}$. Since it is important, we expand $c_\Gamma$:
\begin{multline} \label{eq:2 body matrix elements constrained}
J^{({\Delta},m_{\Delta};d_{\Delta}),({\Lambda},m_{\Lambda};d_{\Lambda})}_{({\Delta'},m_{\Delta'};d_{\Delta'}),({\Lambda'},m_{\Lambda'};d_{\Lambda'})}\\ 
= K^{\Gamma; (d_\Delta, d_\Lambda, \alpha_\Gamma)}_{(d_{\Delta'}, d_{\Lambda'},\beta_{\Gamma})} \left[C^{*\Gamma;\alpha_\Gamma}_{\Delta \Lambda}\right]^{m_\Gamma}_{m_\Delta,m_\Lambda}  \left[C^{\Gamma;\beta_\Gamma}_{\Delta' \Lambda'}\right]^{m_{\Gamma}}_{m_{\Delta'},m_{\Lambda'}}
\end{multline}
Plugging in Eq~[\ref{eq:2 body matrix elements constrained}] into Eq~[\ref{eq:2 body general}], we get the general form of the symmetric 2-spin operator. 
\begin{multline}
\hat{J} = K^{\Gamma; (d_\Delta, d_\Lambda, \alpha_\Gamma)}_{(d_{\Delta'}, d_{\Lambda'},\beta_{\Gamma})} \left[C^{*\Gamma;\alpha_\Gamma}_{\Delta \Lambda}\right]^{m_\Gamma}_{m_\Delta,m_\Lambda}  \left[C^{\Gamma;\beta_\Gamma}_{\Delta' \Lambda'}\right]^{m_{\Gamma}}_{m_{\Delta'},m_{\Lambda'}} \\  \outerproduct{{\Delta},m_{\Delta};d_{\Delta}}{{\Delta'},m_{\Delta'};d_{\Delta'}} \\ \outerproduct{{\Lambda},m_{\Lambda};d_{\Lambda}}{{\Lambda'},m_{\Lambda'};d_{\Lambda'} }.
\end{multline}
This can be greatly simplified using Eq~[\ref{eq:cg basis action}]
\begin{equation} \label{eq:2 spin symmetric}
\hat{J} = K^{\Gamma; c_\Gamma}_{d_\Gamma} \outerproduct{\Gamma,m_\Gamma;c_\Gamma}{\Gamma,m_\Gamma;c'_\Gamma}.
\end{equation}
Where, we have once again reintroduced the short hand notation $c_\Gamma$ to denote the fusion channels labeled by compatible $(d_\Delta, d_\Lambda, \alpha_\Gamma)$ to produce $\Gamma$. $K^{\Gamma; (d_\Delta, d_\Lambda, \alpha_\Gamma)}_{(d_{\Delta'},d_{\Lambda'}, \beta_\Gamma)}$ are now the free parameters. In this form, we can clearly see the similarity with the symmetric 1-spin operator Eq~[\ref{eq:1 body symmetric}]. This helps us see the general picture with arbitrary n-local operators. In a fully reduced basis of the n-spin Hilbert space, matrix elements of symmetric operators can only act on the outer multiplicity space of each irrep.

\section{Detecting the irrep of the eigenstates} \label{app:detecting irrep}
In this appendix, we give details of how we determine the irrep a given eigenstate transforms as. If $U(g) = \bigotimes_{i=1}^L V_i(g)$ is the many-body representation of the on-site symmetry, we want to find out the irrep $\Gamma$ that an eigenstate $\ket{\epsilon}$ or more generally, a set of degenerate eigenstates $\{\ket{\epsilon}_a\}$ transform as 
\begin{equation}
U(g) \ket{\epsilon}_a = \sum\limits_{b=1}^{|\Gamma|} \Gamma(g)_{ab} \ket{\epsilon}_b
\end{equation}
This is an easy task for the symmetry group $SU(2)$, where, all we need to do is operate the eigenstates with the total angular momentum operator, 
\begin{eqnarray}
\vec{S}_{tot}^2 &=& \sum\limits_{a=x,y,z} \left(\sum\limits_{i=1}^{L} S^a_i \right) \left(\sum\limits_{i=1}^{L} S^a_i \right) \\
\vec{S}_{tot}^2 \ket{\epsilon} &=& j (j+1) \ket{\epsilon}
\end{eqnarray}
This way, given an eigenstate that transforms as an irrep of $SU(2)$, we can extract the quantum number $j$ which labels the irrep. For general finite groups, we are not aware of an equivalent technique. However, we now present a strategy that works well for the group $S_3$ in the form of the following theorem. 
\begin{theorem}
Given a normalized vector $\ket{\epsilon}$, $\innerproduct{\epsilon}{\epsilon}=1,$ that transforms as some irreducible representation of $S_3$, $\Gamma = \mathbf{1}, \mathbf{1'}$ or $\mathbf{2}$, we can determine $\Gamma$ using the following two real numbers,
\begin{eqnarray}
\mathcal{X} &=& ~~~~~~\innerproduct{\epsilon|U(x)}{\epsilon}, \nonumber\\
\mathcal{A} &=& Real(\innerproduct{\epsilon|U(a)}{\epsilon}) \nonumber,
\end{eqnarray}  
 where, $a$ and $x$ are the two generators of $S_3$ with presentation $\innerproduct{a,x}{a^3=x^2=1, xax = a^{-1}}$. Specifically,
\begin{eqnarray}
\mathcal{A} = -0.5 &\implies&  \Gamma = \mathbf{2} \nonumber \\
\left(\mathcal{A}, \mathcal{X} \right) = \left(1,~~1 \right) &\implies&  \Gamma = \mathbf{1} \nonumber \\
\left(\mathcal{A}, \mathcal{X} \right) = \left(1,-1 \right) &\implies&  \Gamma = \mathbf{1'} \nonumber 
\end{eqnarray}
\end{theorem}
\begin{proof}
Let us first consider the case when $\ket{\epsilon}$ is some vector in the 2D irrep $\mathbf{2}$. We can expand this vector in the orthonormal eigenbasis of the generator $a$ of the $\mathbf{2}$ representation, $\ket{\omega}, \ket{\omega^*}$,  with eigenvalues $\omega, \omega^*$ respectively, where $\omega = e^{\frac{2 \pi i}{3}}$ (see Appendix.~\ref{app:building symmetric ham}).
\begin{equation}
\ket{\epsilon} = \cos\theta \ket{\omega} + \sin\theta \ket{\omega^*}. \\
\end{equation}
Acting on it by $U(a)$,
\begin{eqnarray}
U(a) \ket{\epsilon} &=& \omega \cos\theta \ket{\omega} + \omega^* \sin\theta \ket{\omega^*} \nonumber\\
\innerproduct{\epsilon|U(a)}{\epsilon} &=& \omega \cos^2\theta + \omega^* \sin^2\theta  \nonumber\\
Real(\innerproduct{\epsilon|U(a)}{\epsilon})&=& Real(\omega) (cos^2\theta + \sin^2\theta) \nonumber\\
\implies \mathcal{A} &=& Real(\omega) = -0.5 \nonumber 
\end{eqnarray}
Let us now consider the case when $\ket{\epsilon}$ transforms as either of the 1D irreps. Since the representation of the generator $a$ is simply $1$ for both 1D irreps, we simply have 
\begin{eqnarray}
U(a) \ket{\epsilon} = \ket{\epsilon}\\
\implies \mathcal{A} = \innerproduct{\epsilon|U(a)}{\epsilon} = 1
\end{eqnarray}
Thus, we can see that $\mathcal{A}$ can separate the 2D irrep $\mathbf{2}$ form the 1D irreps. Also, $\mathcal{A}=-0.5$ is necessary and sufficient for $\ket{\epsilon}$ to transform as $\mathbf{2}$. Furthermore, if $\mathcal{A}=1$, we can determine which 1D irrep $\ket{\epsilon}$ transforms as by considering the transformation under the generator $x$ whose representation is $\pm 1$ for $\Gamma$ being $\mathbf{1}$ or $\mathbf{1'}$ respectively. Thus, if $\mathcal{A} = 1$, 
\begin{eqnarray}
U(x)\ket{\epsilon} = \pm \ket{\epsilon} \nonumber \\
\implies \mathcal{X} = \innerproduct{\epsilon|U(x)}{\epsilon} = \pm1 \nonumber
\end{eqnarray}
implies $\ket{\epsilon}$ transforms as $\mathbf{1}$ or $\mathbf{1'}$ respectively. This concludes the proof.
\end{proof}

In our numerics, we diagonalize our Hamiltonian in the 1D irrep sector by projecting it into the appropriate basis. As discussed above, this means that we need to restrict to the basis states that are left invariant under the action of $U(a)= \bigotimes_{i=1}^{L} V(a)_i$. To see how this is done, consider the action of operator $U(a)$ on a many-body basis state labeled by $S^z$ eigenvalues on each spin
\begin{multline}
U(a) \ket{m_1,m_2,\ldots,m_L} = \omega^{m_1+m_2+\ldots+m_L}\ket{m_1,m_2,\ldots,m_L} \\ = \omega^{S^z_{tot}}\ket{m_1,m_2,\ldots,m_L}
\end{multline}
We need $\omega^{S^z_{tot}} = 1$ which means
\begin{equation} \label{eq:Sz=0mod3}
S^z_{tot} = \sum_{i=1}^L m_i = 0(mod~3).
\end{equation} 
Since $3^{L-1}$ of the $3^L$ basis states satisfy the condition of Eq~.\ref{eq:Sz=0mod3}, this helps us diagonalize an $L$ site Hamiltonian for the price of $L-1$.

\section{Spin glass diagnostics for $S_3$ subgroups}\label{app:SG}
In this appendix, we give details of how the spin-glass (SG) order parameter used in Sec.~\ref{sec:ea} was constructed and also numerical evidence for the assertion that the high disorder region at $\kappa=0$ does not in any form spontaneously break the $S_3$ symmetry. 

First, let us list the elements of $S_3=\{1,a,a^2,x,x a,x a^2\}$ and its five subgroups written in terms of the generators $a,x$ defined in Sec.~\ref{sec:model} and Appendix.~\ref{app:building symmetric ham}:
\begin{enumerate}
	\item $\mathbb{Z}_3 = \{1,a,a^2\}$,
	\item $\mathbb{Z}_{2A} = \{1,x\}$,
	\item $\mathbb{Z}_{2B} = \{1,xa\}$,
	\item $\mathbb{Z}_{2C} = \{1,xa^2\}$,	
	\item $\{1\}$.
\end{enumerate}
For each subgroup $H \subset G$, we design a SG diagnostic that detects SSB of $G \rightarrow H$ and takes the form 
\begin{multline}
\chi^{SG}_H = \frac{1}{L-1} \sum_{i \neq j=1}^{L}  |\innerproduct{\epsilon|\mathcal{O}_{H;i} \mathcal{O}_{H;j}}{\epsilon} -  \innerproduct{\epsilon|\mathcal{O}_{H;i}}{\epsilon} \innerproduct{\epsilon|\mathcal{O}_{H;j}}{\epsilon}|^2.
\end{multline}
 $\mathcal{O}_H$ are local Hermitian order parameters that are chosen to have the following properties under symmetry transformation by $U(g) = \bigotimes_{i=1}^L V(g)_i$
\begin{enumerate}
	\item $\mathcal{O}_H$ transforms trivially under $H$:\\ $U(g)\mathcal{O}_H U^\dagger(g)$ = $\mathcal{O}_H$, $\forall g \in H$.
	\item $\mathcal{O}_H$ transforms non-trivially under $G/H$:\\ $U(g)\mathcal{O}_H U^\dagger(g)$ $\neq$ $\mathcal{O}_H$, $\forall g \notin H$.	
\end{enumerate}
Note that $\chi^{SG}_H$ is invariant under the redefinition $\mathcal{O}_H \rightarrow \mathcal{O}_H + \xi \mathbb{1}$. It can be checked that the following operators satisfy the above properties
\begin{itemize}
	\item $\mathcal{O}_{\mathbb{Z}_3}$ = $S^z$
	\item $\mathcal{O}_{\mathbb{Z}_{2A}}$ = $S^z S^x + S^x S^z$ = $\frac{1}{\sqrt{2}}
	\begin{pmatrix}
		0 & 1 & 0\\
		1 & 0 & -1 \\
		0 & -1 & 0
	\end{pmatrix}$
	\item $\mathcal{O}_{\mathbb{Z}_{2B}}$ = $V(a)\mathcal{O}_{\mathbb{Z}_{2A}}V^\dagger(a)$ = $\frac{1}{\sqrt{2}}
	\begin{pmatrix}
	0 & \omega & 0\\
	\omega^* & 0 & -\omega \\
	0 & -\omega^* & 0
	\end{pmatrix}$
		\item $\mathcal{O}_{\mathbb{Z}_{2C}}$ = $V^\dagger(a)\mathcal{O}_{\mathbb{Z}_{2A}}V(a)$ = $\frac{1}{\sqrt{2}}
	\begin{pmatrix}
	0 & \omega^* & 0\\
	\omega & 0 & -\omega^* \\
	0 & -\omega & 0
	\end{pmatrix}$
	\item $\mathcal{O}_{\{1\}}$ = $S^x$
\end{itemize}

Note that it is important to make sure that the disconnected part of the two point correlation function, $\innerproduct{\epsilon|\mathcal{O}_{H;i}}{\epsilon} \innerproduct{\epsilon|\mathcal{O}_{H;j}}{\epsilon}$  is subtracted when constructing $\chi^{SG}_H$. In previous work like Ref~\cite{bardarson2014}, the local order parameter $\sigma^z$ transformed as a non-trivial irrep of $G/H$  ($\ztwo$ in their case). It is then automatically guaranteed that $\innerproduct{\epsilon}{\sigma^z_i|\epsilon}=0$. Similarly, the local order parameter used in the main text to detect $S_3 \rightarrow \zthree$ SSB, $S^z_i$ transforms as a non-trivial irrep of $S_3/\zthree \cong \ztwo$ which also ensures $\innerproduct{\epsilon}{S^z_i|\epsilon} = 0$ and hence we leave out of the definition of $\ea$. This would not be true if we used a different $\mathcal{O}_{\mathbb{Z}_3}$ like $S^z + \xi (S^z)^2$ or even $S^z + \xi \mathbb{1}$ both of which are equally good to detect $S_3 \rightarrow \zthree$ SSB but would need subtraction of the disconnected part. Similarly, the other SG order parameters we used namely $\mathcal{O}_{\mathbb{Z}_{2A/B/C}}$ and $\mathcal{O}_{\{1\}}$ also need subtraction. 

Fig.~\ref{fig:ea mean all} shows the different SG diagnostics as a function of $\lambda$ for $\kappa=0,1$ averaged over eigenstates across disorder realizations. It can be seen that only $\overline{\ea}$, which detects SSB $S_3 \rightarrow \zthree$ approaches a value that increases with system size in the region discussed in Sec.~\ref{sec:ea}. The value of other SG diagnostics becomes increasingly smaller or approaches a constant value with system size for all  $\lambda$ and $\kappa$ indicating that SSB to that residual subgroup has not taken place in the eigenstates. 

 \begin{figure}[htb]
	\centering
	\begin{tabular}{cc}		
		\includegraphics[width=43mm]{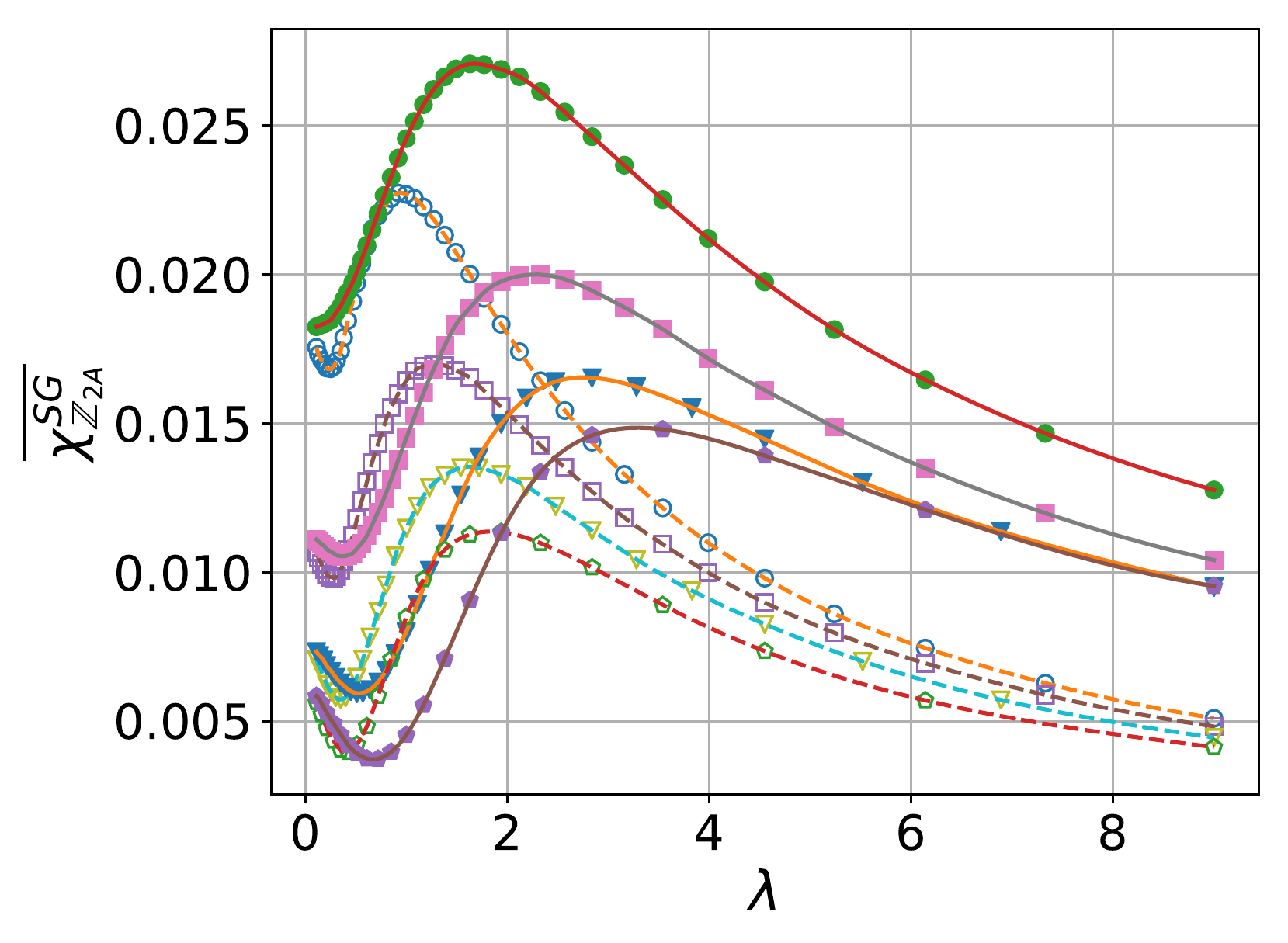}	&
		\includegraphics[width=43mm]{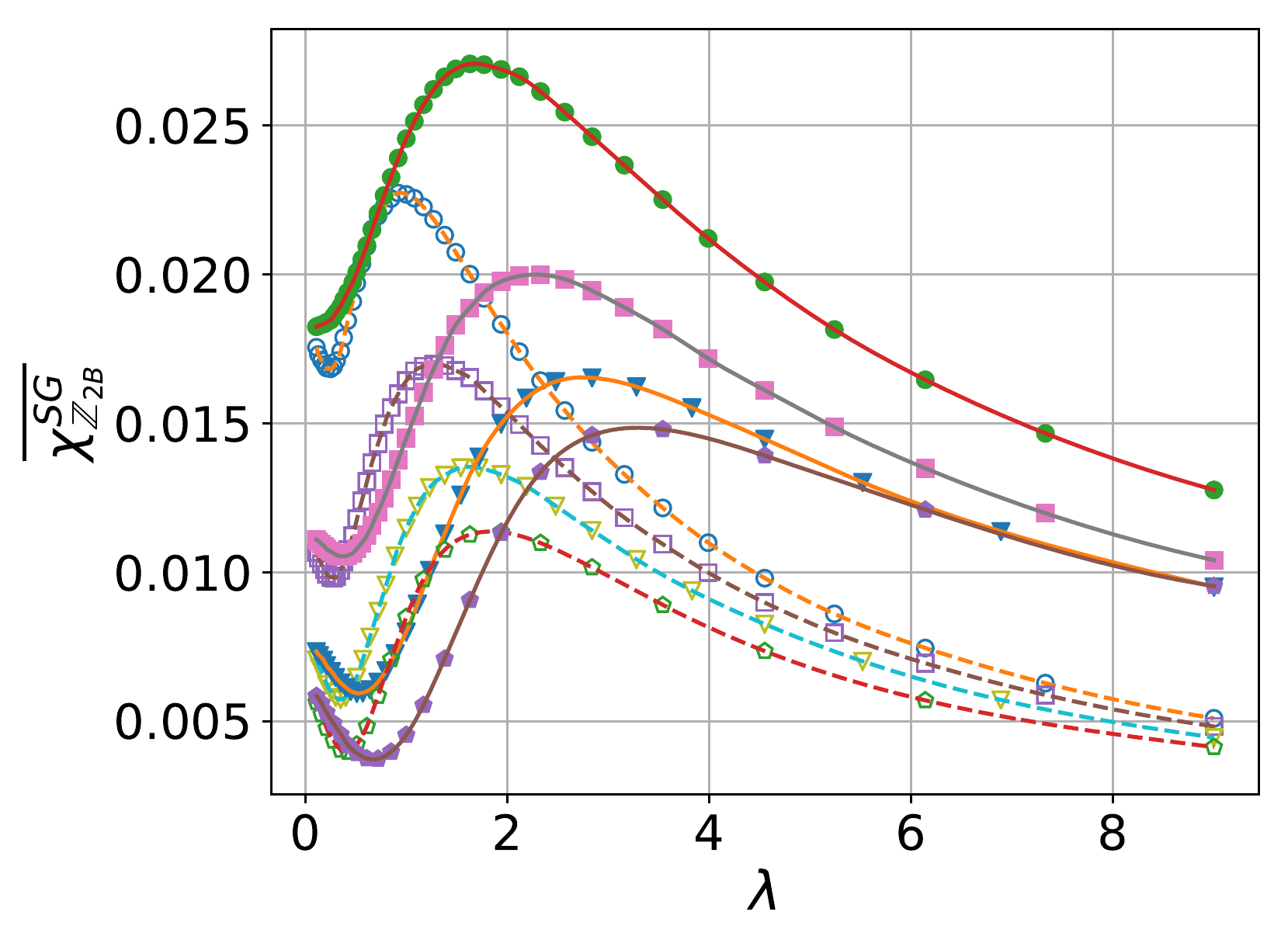}	\\
		\includegraphics[width=43mm]{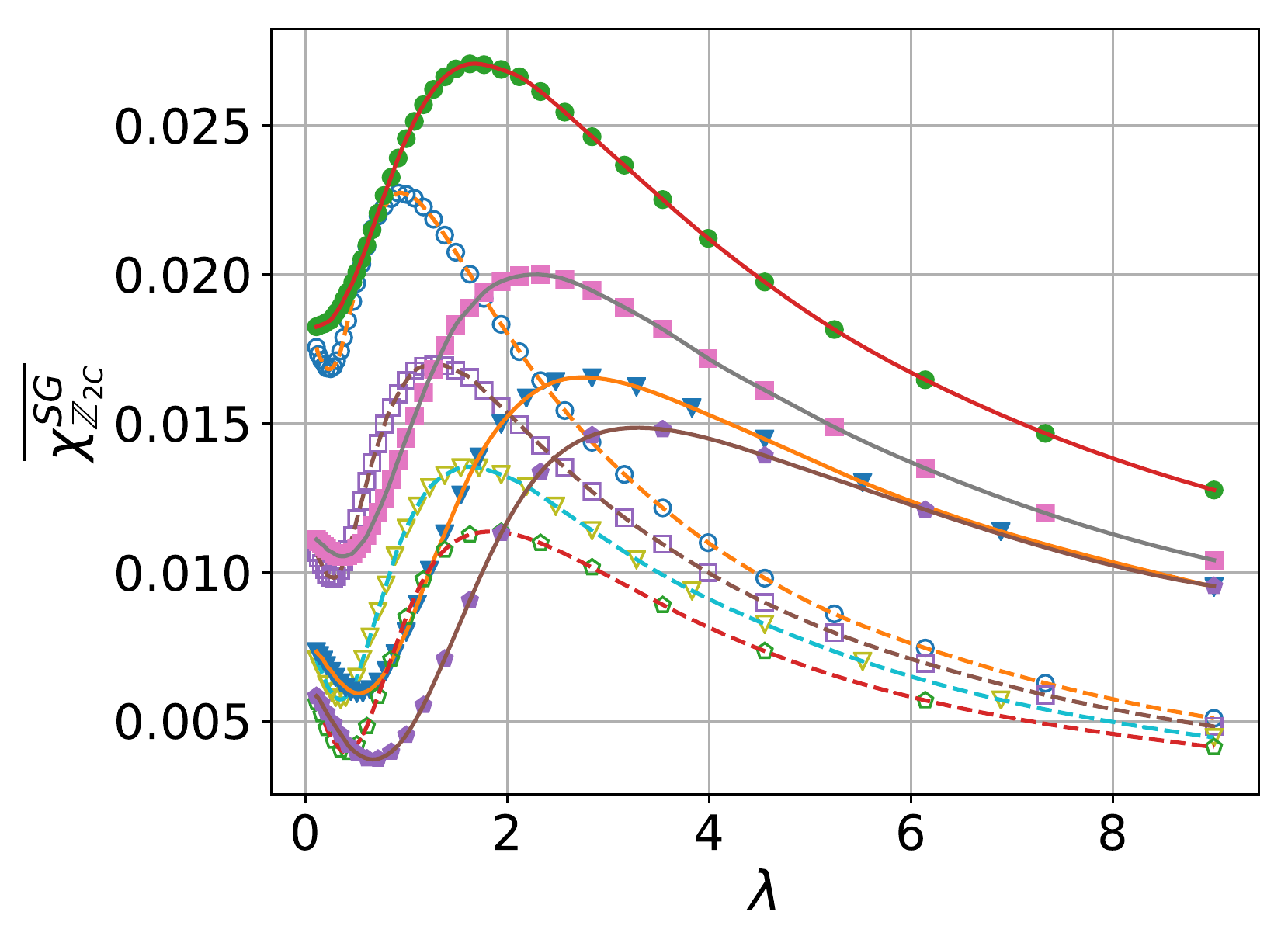}	&
		\includegraphics[width=43mm]{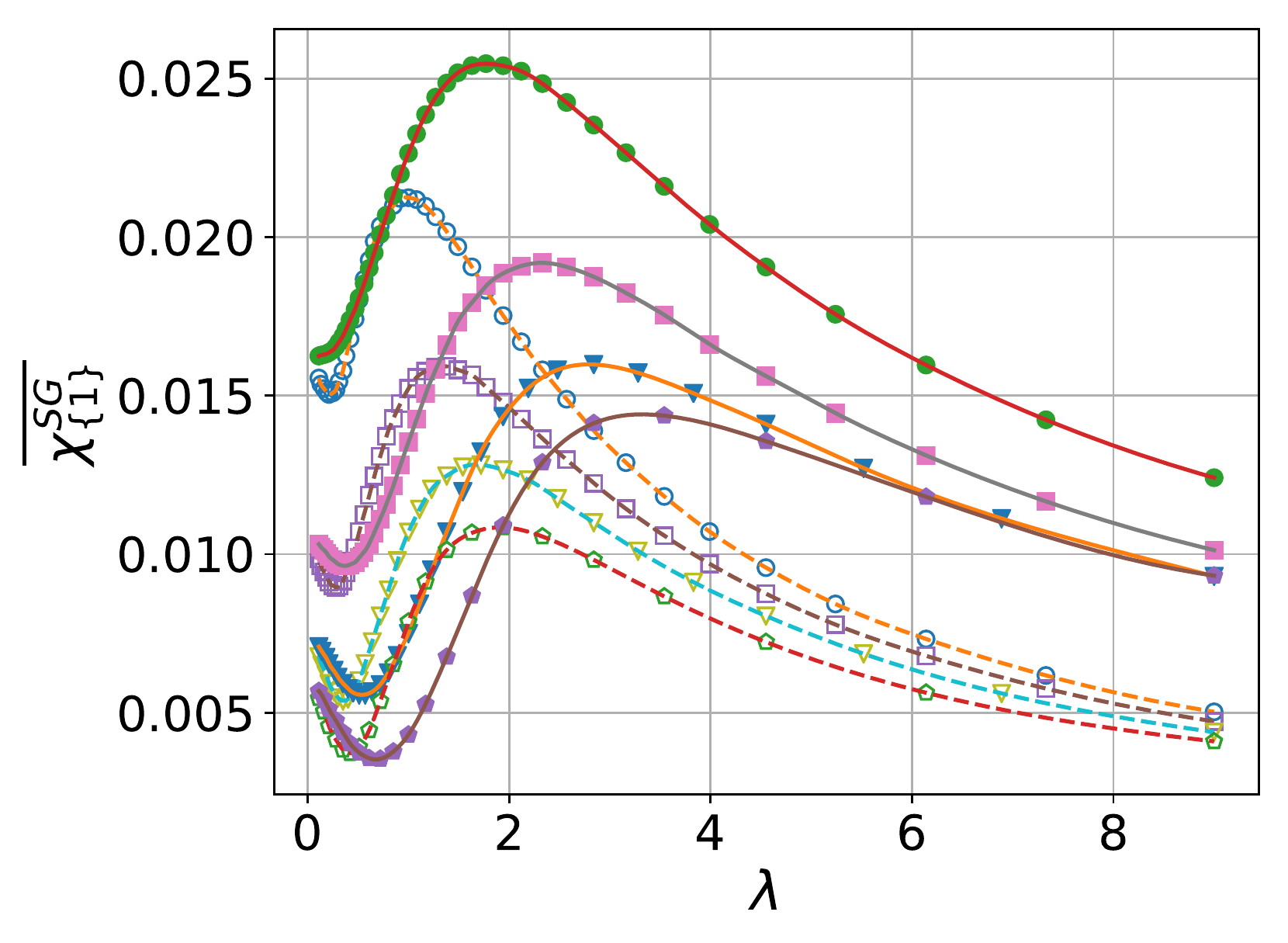} \\
		\multicolumn{2}{l}{\includegraphics[width=80mm]{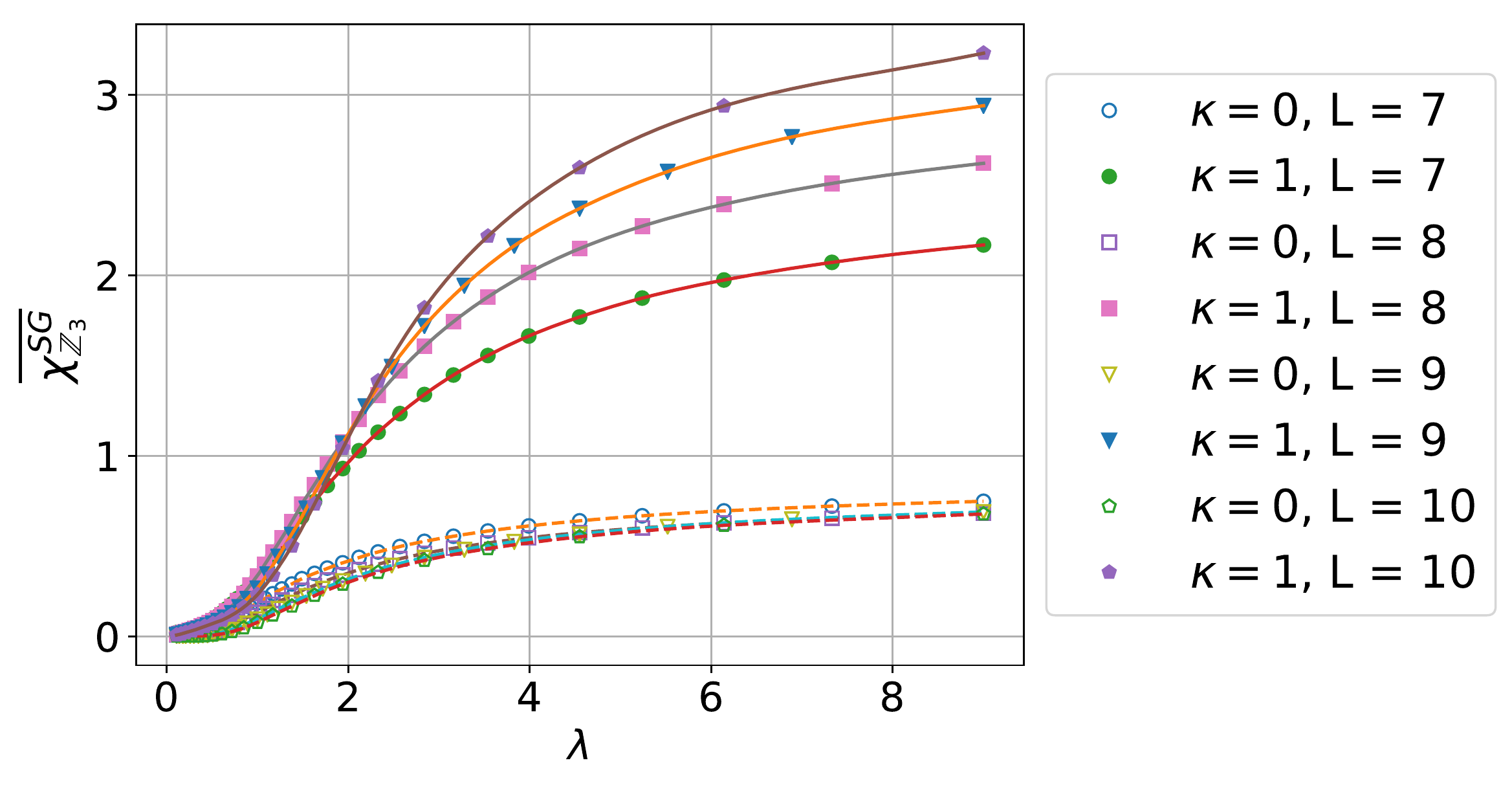}}		
	\end{tabular}
	\caption{ (Color online) SG diagnostics for different subgroups versus $\lambda$ for $\kappa = 0$  and $\kappa= 1$ with spline fit (solid for $\kappa=1$, dashed for $\kappa=0$). 243 eigenstates per disorder realization that transform as 1D irreps sampled for 800 (7,8 sites), 879 (9 sites) and  715 (10 sites) disorder realizations respectively.  The plot for $\overline{\ea}$ is also shown in the main text. \label{fig:ea mean all}}	 					
\end{figure}

\end{document}